\documentclass[12pt]{article}

\usepackage{amsmath}
\usepackage{amssymb}
\usepackage{amsbsy}
\usepackage{amsthm}
\usepackage{mathtools}
\usepackage{mathbbol}
\usepackage{mathrsfs}
\usepackage{bbm} 
\usepackage{graphicx,psfrag,epsf}
\usepackage{enumerate}
\usepackage{natbib}
\usepackage{url} 
\usepackage{float}
\usepackage[skip=5pt]{caption}
\usepackage{subcaption}
\usepackage[normalem]{ulem}
\usepackage{comment}
\usepackage{hhline}
\usepackage{afterpage}
\usepackage{multirow}
\usepackage{adjustbox}
\usepackage{textcomp}
\usepackage{rotating}
\usepackage{tabularx}
\usepackage{xspace}
\usepackage{xr-hyper}
\usepackage[colorlinks,citecolor=blue,urlcolor=blue,filecolor=blue,backref=page]{hyperref}
\usepackage{color,colortbl,soul}
\usepackage[dvipsnames]{xcolor}
\usepackage{blkarray}
\usepackage{tikz}
\usepackage[boxed, vlined]{algorithm2e}
\usepackage[T1]{fontenc}
\usepackage{arydshln}
\usepackage{scalerel,stackengine}
\usepackage{cases}

\usepackage[nobottomtitles*]{titlesec}

\usepackage{booktabs}

\makeatletter
\g@addto@macro\@floatboxreset{\centering}
\makeatother

\floatstyle{plain}
\restylefloat{table}

\SetAlCapSkip{.5em}

\definecolor{Gray}{gray}{0.9}

\newtheorem{theorem}{Theorem}

\allowdisplaybreaks

\makeatletter
\@ifundefined{figuresize}{%
  \newcommand{\figuresize}[1]{\def\@figscale{#1}}%
  \def\@figscale{1}%
}{}
\@ifundefined{figurebox}{%
  \newcommand{\figurebox}[3]{\@ifnextchar[{\@figbox}{}}%
  \def\@figbox[#1]{\centering\includegraphics[width=0.85\linewidth]{#1}}%
}{}
\makeatother

\makeatletter
\renewcommand*\env@matrix[1][*\c@MaxMatrixCols c]{%
	\hskip -\arraycolsep
	\let\@ifnextchar\new@ifnextchar
	\array{#1}}
\makeatother

\newcommand*{\jmlr}[1]{\small\texttt{\url{#1}}\normalsize}

\RequirePackage{doi}

\newcommand{\iidsim}{\overset{iid}{\sim}} 

\stackMath
\newcommand\reallywidehat[1]{%
	\savestack{\tmpbox}{\stretchto{%
			\scaleto{%
				\scalerel*[\widthof{\ensuremath{#1}}]{\kern-.6pt\bigwedge\kern-.6pt}%
				{\rule[-\textheight/2]{1ex}{\textheight}}
			}{\textheight}%
		}{0.5ex}}%
	\stackon[1pt]{#1}{\tmpbox}%
}

\theoremstyle{definition}

\newtheorem{definition}{Definition}[section]

\makeatletter
\newcommand*{\addFileDependency}[1]{
  \typeout{(#1)}
  \@addtofilelist{#1}
  \IfFileExists{#1}{}{\typeout{No file #1.}}
}
\makeatother

\def\mathcolor#1#{\@mathcolor{#1}}
\def\@mathcolor#1#2#3{%
	\protect\leavevmode
	\begingroup
	\color#1{#2}#3%
	\endgroup
}

\renewcommand{\tilde}[1]{\widetilde{#1}}

\newcommand{\bolds}[1]{\boldsymbol{#1}}

\newcommand{\calD}{{\cal D}}

\newcommand{\calS}{{\cal S}}

\newcommand{\calV}{{\cal V}}

\newcommand{\calY}{{\cal Y}}

\newcommand{\sigY}{\calY_o}
\newcommand{\sigV}{\calV_o}

\newcommand{\others}{o}

\newcommand{\Yothers}{\bY_{\others}}
\newcommand{\Vothers}{\bV_{\others}}

\newcommand{\RRd}{\mathbb{R}^d}
\newcommand{\intRd}{\int_{\RRd}}

\newcommand{\covf}{\text{cov}}

\newcommand{\bA}{\bolds{A}}

\newcommand{\bB}{\bolds{B}}

\newcommand{\bD}{\bolds{D}}

\newcommand{\bg}{\bolds{g}}
\newcommand{\bG}{\bolds{G}}
\newcommand{\bh}{\bolds{h}}

\newcommand{\bI}{\bolds{I}}

\newcommand{\bK}{\bolds{K}}
\newcommand{\bl}{\bolds{\ell}}
\newcommand{\bL}{\bolds{L}}

\newcommand{\bP}{\bolds{P}}
\newcommand{\bQ}{\bolds{Q}}
\newcommand{\br}{\bolds{r}}
\newcommand{\bR}{\bolds{R}}
\newcommand{\bs}{\bolds{s}}
\newcommand{\bS}{\bolds{S}}

\newcommand{\bU}{\bolds{U}}
\newcommand{\bv}{\bolds{v}}
\newcommand{\bV}{\bolds{V}}
\newcommand{\bw}{\bolds{w}}
\newcommand{\bW}{\bolds{W}}

\newcommand{\by}{\bolds{y}}
\newcommand{\bY}{\bolds{Y}}

\newcommand{\bZ}{\bolds{Z}}

\newcommand{\vecop}{\text{vec}}

\newcommand{\bzero}{\mathbf{0}}

\newcommand{\bSigma}{\bolds{\Sigma}}

\newcommand{\bGamma}{\bolds{\Gamma}}

\newcommand{\bLambda}{\bolds{\Lambda}}

\newcommand{\bomega}{\bolds{\omega}}

\usepackage{algorithmicx}
\usepackage{algpseudocode}

\usepackage{tcolorbox}
\newtcolorbox[auto counter]{reviewcommentinside}[1][]{box align=center,
    width=0.9\textwidth,
    colframe = teal,
    colback=teal!10,
    code={\spacingset{0.9}},
    #1}

\newcommand{\indep}{\mathrel{\perp\!\!\!\perp}}

\usepackage{tikz}
\usetikzlibrary{bayesnet}

\begin{document}

\def\spacingset#1{\renewcommand{\baselinestretch}%
{#1}\small\normalsize} \spacingset{1}

\date{}

\newcommand{\footremember}[2]{%
    \footnote{#2}
    \newcounter{#1}
    \setcounter{#1}{\value{footnote}}%
}
\newcommand{\footrecall}[1]{%
    \footnotemark[\value{#1}]%
} 
\newcommand{\mytitle}{Partial correlation networks \\of Gaussian processes}  

\title{\bf \mytitle}
  \author{Michele Peruzzi\footremember{alley}{\url{peruzzi@umich.edu} Department of Biostatistics, University of Michigan--Ann Arbor.}\\
  }
  \maketitle
  
\bigskip
\begin{abstract}
In Gaussian graphical models, conditional independence and partial correlations are natural inferential targets for understanding direct relationships in multivariate data. No comparable framework exists for spatial processes, where multivariate analysis defaults to modeling unconditional cross-covariance structure, even when direct relationships remain of scientific interest. We address this gap by establishing a novel characterization of process-level partial correlation for multivariate Gaussian processes that recovers a direct link with Gaussian graphical models. Our analysis proceeds through a class of stationary multivariate processes, termed spectrally inside-out, in which a precision matrix modulates the strength of conditional dependence and yields necessary and sufficient conditions for conditional independence. Within this class, partial cross-correlation functions factorize into a process-level partial correlation coefficient and an attenuation term independent of cross-process parameters. The spectrally inside-out class includes the separable coregionalization model, a process convolution construction, and the parsimonious multivariate Mat\'ern, for which such a characterization was previously thought unavailable. We further show that a nonstationary inside-out model satisfies the same factorization and admits the same necessary and sufficient conditions. Our results clarify the limitations of existing approaches: linear coregionalization models encode conditional independence through the zero pattern of the inverse factor loading matrix and do not result in interpretable partial cross-correlation functions. Low-rank spatial factor models lack a meaningful graphical characterization. Methods that enforce network structure through auxiliary graphical layers only characterize presence or absence of graph edges. We illustrate our results through synthetic and real data.
\end{abstract}

\noindent%
\vfill

\newpage
\spacingset{1.4} 

\section{Introduction}
In many applied sciences, including spatial ``omics,'' community ecology, and environmental sciences, datasets often consist of measurements of multiple spatially indexed variables observed across a domain. These variables may represent different ``species'' broadly defined, such as ecological populations, environmental quantities, molecular markers, or cell types. In these settings, we associate each species with a component of a multivariate spatial process characterizing the joint spatial distribution of all species across the domain. This enables a joint characterization of spatial and cross-species dependence. It is common to assume the multivariate process is Gaussian to ensure tractability; in this case, one estimates dependence by parametrizing the spatial cross-covariance function, which, under a zero-mean assumption, fully characterizes the multivariate process \citep{genton_cross-covariance_2015}. 
Traditional analyses in these settings focus on unconditional dependence---that is, the cross-covariance function between any two component processes marginalized over, rather than conditional on, the remaining processes.
This view risks conflating direct and indirect associations. For example, in spatial proteomics of the tumor microenvironment, cross-covariance analysis may reveal that certain immune and cancer cell populations covary strongly across tissue regions. This unconditional view cannot infer whether such co-variation reflects a direct interaction between the two cell populations or arises indirectly through shared dependence on a third population or on favorable local conditions. In this scenario, we seek to quantify the strength of direct association to understand the global mechanisms that shape the tumor microenvironment and for assessing the effect of therapeutic interventions \citep{najafi_tumor_2019}. Similar needs arise in environmental monitoring, where one may ask whether two pollutants co-vary directly or through shared drivers, and in community ecology, where direct biotic interactions are rarely observed yet one seeks to infer them from spatially referenced occurrence or abundance data \citep{blanchet_co_2020}.

Disentangling direct from indirect associations requires moving from unconditional to conditional measures of dependence, i.e., to a study of partial cross-correlations. In classical multivariate analysis of independent data, Gaussian graphical models provide the standard framework for modeling conditional dependence among variables 
\citep{dempster_covariance_1972, speed_gaussian_1986, lauritzen_graphical_1996, whittaker_graphical_2009}. Partial correlations are obtained by standardizing the precision (inverse covariance) matrix, which encodes pairwise conditional dependence. In particular, a zero partial correlation between two variables implies their conditional independence given all remaining variables. 
Extending partial correlation analysis from finite-dimensional random vectors to infinite-dimensional stochastic processes is nontrivial, as classical Gaussian graphical model results do not transfer directly. 
The spectral and spatial domains offer different tradeoffs in addressing this challenge. For stationary processes, process-level conditional independence corresponds to zeros in the inverse spectral density matrix \citep{dahlhaus_graphical_2000, guinness_multivariate_2014}, but estimation in the spectral domain is less direct and the results are not expressed in terms of spatial lags \citep{grainger_spectral_2025}. In the spatial domain, cross-covariance estimation is well developed, but cross-covariance models have received little attention from the perspective of process-level conditional dependence. Under the implicit assumption that cross-covariance models cannot directly encode such structure, conditional independence can be enforced through an auxiliary graph through Markov combinations of consistent probability measures \citep{zhu_bayesian_2016} or via a spatial ``stitching'' approach \cite{dey_graphical_2022}. These methods are limited to presence or absence of graph edges and do not quantify the sign and strength of conditional dependence.

In this article, we develop a partial correlation analysis of multivariate Gaussian processes. We define the partial cross-correlation function, extending the process-level conditional independence characterization of \cite{dahlhaus_graphical_2000} to a function quantifying the strength of spatial conditional dependence. To characterize this function explicitly, we introduce a class of stationary multivariate processes, which we term \emph{spectrally inside-out}, for which the partial cross-correlation function admits a clean factorization into a graph-level quantity and a spatial attenuation function. In this class of models, the spectral density matrix factors into component-specific marginal contributions and a shared $q\times q$ cross-dependence matrix $\bSigma$. The inverse $\bQ=\bSigma^{-1}$ encodes conditional dependence at the process level. In particular, zeros in $\bQ$ are necessary and sufficient for conditional independence of the corresponding component processes; more generally, the entries of $\bQ$ modulate the strength of conditional cross-dependence between processes. 

We show that the spectrally inside-out class includes the separable coregionalization model, a process convolution construction which builds on \cite{ver_hoef_constructing_1998, gaspari_construction_1999, higdon_space_2002, majumdar_multivariate_2007, alvarez_computationally_2011}, and the parsimonious specification of the multivariate Mat\'ern \citep{gneiting_matern_2010, apanasovich_valid_2012, emery_new_2022, yarger_multivariate_2024}. 
For the multivariate Mat\'ern, \cite{dey_graphical_2022} had concluded that it lacks a nontrivial parametrization yielding process-level conditional independence; we establish that its parsimonious specification does in fact encode conditional independence and, as a spectrally inside-out process, also supports a fully interpretable partial cross-correlation function. We also prove that the recently introduced inside-out cross-covariance model \citep{peruzzi_iox}, which is not inside-out in the spectral sense, admits an analogous factorization of the partial cross-correlation function and the same necessary and sufficient conditions for process-level conditional independence based on $\bQ$.

We contrast these results with the linear model of coregionalization \citep{matheron82, wackernagel03, schmidtgelfand, alie24}, which does not belong to the spectrally inside-out class. In this case, necessary and sufficient conditions for process-level conditional independence are expressed through the zero pattern of the inverse factor loading matrix, zeros in the inverse coregionalization matrix are not sufficient for process-level conditional independence, and the partial cross-correlation does not factor into graph-level and spatial components as it does in spectrally inside-out models. Therefore, linear coregionalization, like the auxiliary graphical constructions discussed above, can only encode presence or absence of edges and does not yield interpretable partial cross-correlation functions. Spatial factor models \citep{taylor2019spatial, zhangbanerjee20} lack a useful graphical characterization altogether.

After defining process-level quantities in Section \ref{sec:defs}, we present our partial correlation analysis of spectrally inside-out processes in Section \ref{sec:sio}; Section \ref{sec:iox} extends the same characterization to the inside-out cross-covariance model, whereas Section \ref{sec:lmc} considers the linear model of coregionalization. Illustrations on synthetic and real data are in Section \ref{sec:illu}. We conclude with a discussion. The supplement includes all proofs as well as additional details on the applications. Code to reproduce all analyses is available at \url{github.com/mkln/gp-pcorr}.

\section{Partial correlations for multivariate Gaussian processes}\label{sec:defs}
Consider a $q$-variate Gaussian process $\{\bY(\bl):\bl \in \calD \} = \bY(\cdot)$ defined in the spatial domain $\calD \subset \RRd$. Let $\bSigma = [\sigma_{ij}]$ be a covariance matrix with inverse $\bSigma^{-1} = \bQ = [Q_{ij}]$, and let $\bLambda$ be a full-rank $q\times q$ matrix such that $\bLambda \bLambda^T = \bSigma$.
Taking the component processes $y_i(\cdot)$ and $y_j(\cdot)$ of $\bY(\cdot)$, we define $\bY_{\others}(\cdot) = \{ y_{r}(\cdot) \mid r \in \{1, \dots, q\} \setminus \{i,j\} \}$, and the residual processes as $z_i(\bl) = y_i(\bl) - E(y_i(\bl) \mid \sigY )$, similarly for $z_j(\cdot)$, where $\sigY = \sigma\{ \bY_{\others}(\bl), \bl \in\calD  \}$, and where $\sigma\{\cdot \}$ is the sigma algebra defined by its argument. We assume $\text{var}\{ z_i(\bl) \}>0$ for all $\bl\in \calD$, similarly for $z_j(\cdot)$; that is, we exclude that $y_i(\cdot)$ or $y_j(\cdot)$ are deterministic functions of $\bY_{\others}(\cdot)$. We call $\covf\{y_i(\bl), y_i(\bl') \}$ the marginal covariance of the $i$th process, $\covf\{y_i(\bl), y_j(\bl') \}$ the cross-covariance between the $i$th and $j$th process, and similarly for the cross-correlation. 
In the following definitions, we formalize partial cross-correlation as the cross-correlation of the residual processes, then use partial cross-correlation to define process-level conditional independence. 
\begin{definition}\label{def:pscf}
The partial cross-correlation function between $y_i(\cdot)$ and $y_j(\cdot)$ is 
\[ \text{corr}\{y_i(\bl), y_j(\bl') \mid \sigY \}=\text{corr}\{ z_i(\bl), z_j(\bl') \} = \frac{\covf\{z_i(\bl), z_j(\bl')\} }{ ( \text{var}\{z_i(\bl)\} \text{var}\{z_j(\bl')\})^{1/2}}.\]
\end{definition}
\begin{definition}\label{def:plci}
The processes $y_i(\cdot)$ and $y_j(\cdot)$ are conditionally independent given $\bY_{\others}(\cdot)$  
if $\text{corr}\{y_i(\bl), y_j(\bl') \mid \sigY \}=0$ for all $\bl, \bl'\in \calD$. We write $y_i(\cdot) \indep y_j(\cdot) \mid \bY_{\others}(\cdot)$.
\end{definition}
\cite{dahlhaus_graphical_2000} defines process-level conditional independence using $\covf\{z_i(\bl), z_j(\bl') \}=0$ for all $\bl,\bl'\in\calD$; the two definitions are equivalent under the assumption of nonzero variances of the residual processes. The partial cross-correlation function quantifies the sign and strength of conditional dependence between processes and offers a more complete view than a conditional independence graph, which only encodes presence or absence of edges. 
Unfortunately, the conditional view is traditionally unavailable in statistical analyses of multivariate spatial dependence, which parametrize the covariance function $\covf\{y_i(\bl), y_j(\bl')\}$ for each $i,j$ and do not typically yield an explicit parametric form for $\text{corr}\{y_i(\bl), y_j(\bl') \mid \sigY \}$. In the following, we show that for a class of flexible nonseparable cross-covariance models, the partial correlation function is explicitly available and can be interpreted as easily as the usual cross-covariance function. 

\section{Spectrally inside-out processes}\label{sec:sio}
We define a class of nonseparable, stationary processes whose spectral density matrix factors to separate cross-process dependence from marginal spectral structure. For this class, partial cross-correlations factorize into two interpretable components, enabling a direct characterization of process-level conditional independence via $\bQ$. 

\begin{definition}\label{def:sio}
A $q$-variate stationary Gaussian process $\bY(\cdot)= (y_1(\cdot), \dots, y_q(\cdot))^T$ is \textit{spectrally inside-out} if its spectral density matrix is $\bS_Y(\bomega) = \bD(\bomega)\bSigma \bD(\bomega)^{*}$ for almost 
all $\bomega$, where $\bD(\bomega)^*$ is the conjugate transpose of $\bD(\bomega)=\text{diag}\{d_1(\bomega),\dots,d_q(\bomega)\}$, where each $d_j(\bomega)$ is such that $\intRd |d_j(\bomega)|^2=1$, $d_j(\bomega)\neq 0$ for almost all $\bomega$, and does not depend on elements of $\bSigma$.
\end{definition}
\begin{theorem}\label{thm:sio_partcor}

If $\bY(\cdot)$ is spectrally inside-out, then $\text{cov}\{y_i(\bl), y_j(\bl')\}=\sigma_{ij} c_{ij}(\bl-\bl')$, where $c_{ij}(\cdot)$ does not depend on $\bSigma$ or $\bQ$, and 
\begin{enumerate}
\item[(a)] $\text{corr}\{y_i(\bl), y_j(\bl') \} = \sigma_{ij}/(\sigma_{ii}\sigma_{jj})^{1/2} c_{ij}(\bl - \bl')$,
\item[(b)] $\text{corr}\{y_i(\bl), y_j(\bl') \mid \sigY \} = -Q_{ij}/(Q_{ii}Q_{jj})^{1/2} c_{ij}(\bl - \bl'),$
\item[(c)] $y_i(\cdot) \indep y_j(\cdot) \mid \bY_{\others}(\cdot)$ if and only if $Q_{ij} = 0$.
\end{enumerate}

\end{theorem}
The proof, in the Supplement, proceeds by showing that the spectral density matrix of the residual process factors into a frequency-dependent diagonal part and a frequency-free precision matrix block, and the partial cross-correlation retains the same spatial shape as the marginal cross-correlation, scaled by a function of the corresponding precision matrix entry. 
Theorem~\ref{thm:sio_partcor} establishes that spectrally inside-out models require no additional auxiliary graphical model to characterize conditional independence, as this can be achieved directly via zeros in $\bQ$. Furthermore, this class of models is not limited to characterizing discrete graphical dependence, since it enables a full characterization of partial cross-correlation. The partial cross-correlation function between $y_i(\cdot)$ and $y_j(\cdot)$ given all others is determined by two factors. The first is $c_{ij}(\bl-\bl')$, which is uniquely determined by the spectral similarity between $d_i(\bomega)$ and $d_j(\bomega)$, and attenuates both types of correlation functions. 
The second factor is the process-level partial correlation coefficient $r_{ij}=-Q_{ij}/(Q_{ii} Q_{jj})^{1/2}$. 
This characterization is appealing because $c_{ij}(\bl - \bl')$ involves no cross-process parameters, which implies that all cross-process dependence is encoded in $\bSigma$ (unconditional) and its inverse $\bQ$ (conditional). 
Finally, there is no requirement in Theorem~\ref{thm:sio_partcor} that $c_{ij}(\bzero)=1$. This implies that $\text{cov}\{ \bY(\bl), \bY(\bl) \} = \bSigma \odot \bGamma$ for some positive semi-definite matrix $\bGamma$, where ``$\odot$'' is the Hadamard element-by-element product. This means that the covariance matrix $\bSigma$ is not in general interpreted as the colocated covariance or coregionalization matrix. Similarly, the colocated precision matrix is computed as $(\text{cov}\{ \bY(\bl), \bY(\bl) \})^{-1} = (\bSigma \odot \bGamma)^{-1}$. From the colocated precision matrix we can compute the colocated \textit{point-wise} partial correlation matrix, whose $(i,j)$ element is $\text{corr}\{ y_i(\bl), y_j(\bl) \mid \bY_{\others}(\bl) \}$; this should not be confused with the colocated \textit{process-level} partial correlation matrix, which is the $q\times q$ matrix whose $(i,j)$ element is $\text{corr}\{ y_i(\bl), y_j(\bl) \mid \sigY \} = -Q_{ij}/(Q_{ii}Q_{jj})^{1/2} c_{ij}(\bzero)$, and quantifies the strength of conditional dependence between $y_i(\cdot)$ and $y_j(\cdot)$ at zero distance. The distinction between point-wise and process-level quantities only collapses for the separable model (see Section \ref{sec:sep}).

\begin{figure}[!htb]
\figuresize{.6}
\figurebox{30pc}{}{}[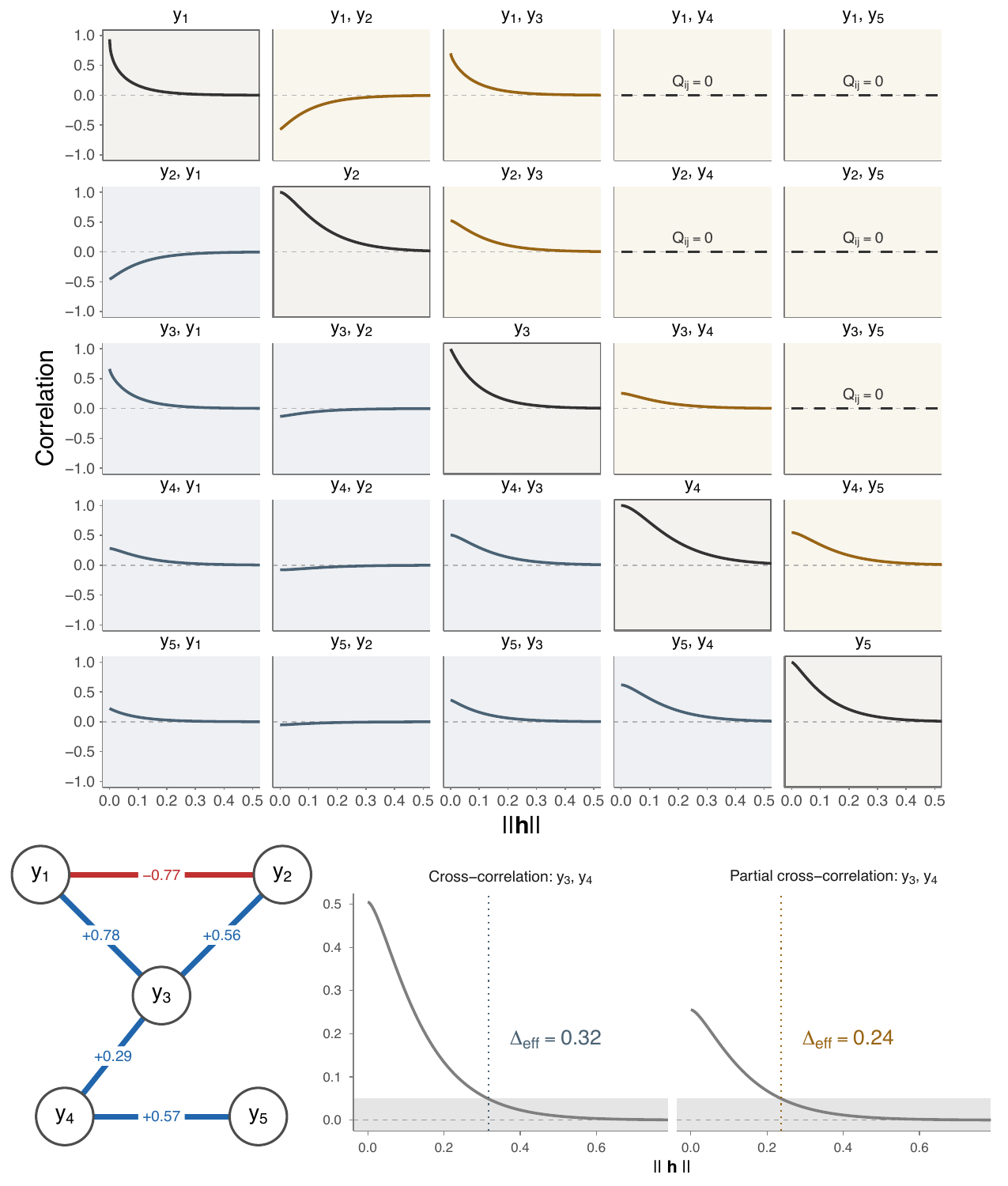]
\caption{Cross-dependence across $q=5$ processes obtained through a parsimonious Mat\'ern model with $\phi=10$, smoothness $\nu \in \{ 0.2, 1, 0.5, 1.4, 0.75 \}$, 
    \emph{Top:} spatial cross-correlation functions (lower triangular panels); spatial marginal correlation (diagonal); partial cross-correlation functions (upper triangular). 
    \emph{Bottom left:} conditional independence graph encoded by $\bQ$; 
    edge labels are process-level partial correlation coefficients ($-Q_{ij}/(Q_{ii}Q_{jj})^{1/2}$).
    \emph{Bottom right:} effective cross-range versus partial effective cross-range 
    between processes $y_3$ and $y_4$; the grey band marks $\leq 0.05$ correlation.}
\label{fig:partial_corr}
\end{figure}    

Our results give rise to new inferential targets at the process level. A traditional analysis of multivariate spatial data is limited to spatial cross-correlation functions (Figure \ref{fig:partial_corr}, lower triangle of the top panel). Methods that introduce an auxiliary graphical model \citep{zhu_bayesian_2016, dey_graphical_2022} additionally yield an unweighted conditional independence graph. Our framework also assigns signs and weights to each edge of the graph via process-level partial correlation coefficients (Figure \ref{fig:partial_corr}, bottom left) and provides the full partial cross-correlation function (upper triangle of the top panel), from which new interpretable summaries emerge. For instance, we can define the \textit{partial effective cross-range} as the distance at which partial cross-correlation decays below a threshold; beyond this range, conditional dependence between two processes is negligible (Figure \ref{fig:partial_corr}, bottom right).
In the example in Figure \ref{fig:partial_corr}, cross-correlation between $y_3(\cdot)$ and $y_5(\cdot)$ is positive and long-ranged, yet these two processes are conditionally independent given the others: their association is entirely mediated by $y_4(\cdot)$. Conversely, $y_3(\cdot)$ and $y_2(\cdot)$ are weakly negatively correlated in the unconditional view, but their partial correlation is strong, positive, and long range. This sign reversal indicates that the direct association between the two processes is obscured by their opposing relationships with $y_3(\cdot)$. Finally, the partial cross-range between $y_3(\cdot)$ and $y_4(\cdot)$ is shorter than the unconditional cross-range: their direct dependence is more spatially localized than an unconditional analysis would suggest. Traditional spatial analysis fails to uncover all these effects. In Section \ref{sec:simulations}, we generate data from the same setup as Figure \ref{fig:partial_corr} and demonstrate that this level of inferential detail is attainable in practice using standard estimation methods, without requiring ad-hoc algorithms.

The next three subsections cover three cross-covariance models which belong to the spectrally inside-out class.

\subsection{Separable coregionalization}\label{sec:sep}
Let $\bY(\cdot) = \bLambda \bW(\cdot)$, where $\bW(\cdot)$ is a $q$-variate Gaussian process whose independent components are all endowed with the same stationary correlation function $\rho(\bl, \bl')=\rho(\bl-\bl')$. Let $\tilde{\rho}(\bomega)$ be the spectral density of $\rho(\cdot, \cdot)$ at frequency $\bomega$, and assume $\tilde{\rho}(\bomega)>0$ for almost all $\bomega\in \RRd$. This gives the intrinsic or separable coregionalization model, which is such that $
\covf\{ y_i(\bl), y_j(\bl') \} = \sigma_{ij} \rho(\bl - \bl').$
\begin{theorem}\label{thm:intrinsic}
Under the intrinsic (separable) coregionalization model for $\bY(\cdot)$,
\begin{enumerate}
\item[(a)] $\text{corr}\{y_i(\bl), y_j(\bl') \} = \sigma_{ij}/(\sigma_{ii}\sigma_{jj})^{1/2} \rho(\bl - \bl')$,
\item[(b)] $\text{corr}\{y_i(\bl), y_j(\bl') \mid \sigY \} = -Q_{ij}/(Q_{ii}Q_{jj})^{1/2} \rho(\bl - \bl')$,
\item[(c)] $y_i(\cdot) \indep y_j(\cdot) \mid \bY_{\others}(\cdot)$ if and only if $ Q_{ij} = 0$
\end{enumerate}
\end{theorem}
The proof is in the Supplement. The conditional independence part of this result was previously noted in Section 3.1 of \cite{dey_graphical_2022}. Theorem~\ref{thm:intrinsic} expands that result further by relating $\bQ$ to the partial cross-correlation of the process. The separable model is restrictive because it forces all outcomes to share a common correlation function. In modern high-dimensional spatial data, it is more realistic to allow at least one outcome to have distinct spatial structure, motivating models that support outcome-specific dependence. Since $\rho(\cdot)$ is a stationary correlation function, $\rho(\bzero)=1$, implying that point-wise and process-level quantities coincide in this model.

\subsection{Process convolutions of correlated white noise}\label{sec:convnew}
Let $\bU(\cdot)$ be $q$-variate Gaussian white noise with uncorrelated components and $\bV(\cdot)=\bLambda \bU(\cdot)$. Let $k_{j}(\cdot)$ be integrable and square-integrable smoothing kernels for $j\in\{1,\dots, q\}$ which we assume to be normalized as $\intRd k_j(r)^2dr=1$. Let $d_{j}(\bomega)$ be the Fourier transform of $k_{j}(\cdot)$, which we assume $d_{j}(\bomega)\neq 0$ for almost all $\bomega$. 
We build $y_j(\cdot)$ via component-specific process convolutions:
\begin{equation}\label{eq:convnew_construct}
\begin{aligned}
y_j(\bl) 
   = \intRd k_j(\bl - \bh) d v_j(\bh) = (k_j \star v_j)(\bl),
    \end{aligned}
\end{equation}
where $v_j(\cdot)$ is the $j$th component of $\bV(\cdot)$. Let $k_{ij}(\bh) = \intRd k_i(\br) \, k_j(\br+\bh) \, d\br$.
The cross-covariance function is computed as $\covf\{y_i(\bl), y_j(\bl')\} = \sigma_{ij} k_{ij}(\bl-\bl')$; it is stationary and nonseparable.

\begin{theorem}\label{thm:convnew}
Under the component-specific process convolution model for $\bY(\cdot)$, 
\begin{enumerate}
\item[(a)] $\text{corr}\{y_i(\bl), y_j(\bl') \} = \sigma_{ij}/(\sigma_{ii}\sigma_{jj})^{1/2} k_{ij}(\bl-\bl')$,
\item[(b)] $\text{corr}\{y_i(\bl), y_j(\bl') \mid \sigY \} = -Q_{ij}/(Q_{ii}Q_{jj})^{1/2} k_{ij}(\bl-\bl')$,
\item[(c)] $y_i(\cdot) \indep y_j(\cdot) \mid \bY_{\others}(\cdot)$ if and only if $ Q_{ij} = 0$.
\end{enumerate}
\end{theorem}
The proof, in the Supplement, casts \eqref{eq:convnew_construct} as a spectrally inside-out process. We interpret this model as independent convolutions of the components of a ``seed'' white noise multivariate process whose components are dependent. 
When $k_{ij}(\bl,\bl')$ is not available in closed form, convolution models may scale poorly to data dimension. A practical alternative is to discretize the model on a finite grid. The discretized model lets $\mathcal{S} = \{ \bs_1, \dots, \bs_{n_S} \}$ be a grid of knots covering the spatial domain,  $y_j(\bl) = \bg_j(\bl)^T \bv_{j},$ where $v_j(\bs_i)$ is the $i$th element of $\bv_{j}$, $\bg_j(\bl) = (k_j(\bl - \bs_1) \Delta_1^{1/2}, \dots,  k_j(\bl - \bs_{n_S}) \Delta_{n_S}^{1/2})^T$, and $\Delta_i$ is the area corresponding to knot $\bs_i$. For any set of locations $\mathcal{L}=\{\bl_1, \dots, \bl_n \}$, we can write in vector form $\by_j = \bG_j \bv_j$, where $\bG_j$ is a $n \times n_S$ matrix whose $i$th row is $\bg_j(\bl_i)$. Then, letting $\bY$ be the $n\times q$ matrix whose $j$th column is $\by_j = (y_j(\bl_1), \dots, y_j(\bl_n))^T$ and $\vecop(\bY)=( \by_1^T\ \cdots \ \by_q^T )^T$,
\begin{equation}\label{eq:discconv_S}
\covf\{ \vecop(\bY) \} \approx \text{diag}\{\bG_1, \dots, \bG_q\} \ (\bSigma \otimes \bI_{n_S}) \ \text{diag}\{\bG_1^T, \dots, \bG_q^T\}.
\end{equation}
The discretized formulation in \eqref{eq:discconv_S} provides a fixed-rank, nonstationary approximation to \eqref{eq:convnew_construct} which yields substantial computational advantages if $n_S \ll n$. As a result, the convolution approach is most advantageous in settings involving smooth spatial processes, where such low-rank approximations remain accurate. More broadly, the highly structured covariance matrix in \eqref{eq:discconv_S} suggests that this model is well suited to the high-dimensional setting where graphical modeling is most informative: all cross-process dependence enters through $\bSigma$, so sparsity in $\bQ$ can be estimated or penalized directly, and the conditional independence graph identified, without modifying the spatial components of the model.

\subsection{Multivariate Mat\'ern model}
Building on the connection between process convolutions and Mat\'ern correlations noted by \citet{gneiting_matern_2010}, one may expect a characterization analogous to Theorem~\ref{thm:convnew} to emerge for multivariate Mat\'ern covariances under appropriate constraints on parameter space: we show this is indeed the case for the parsimonious multivariate Mat\'ern. We parametrize the model as
\begin{equation}\label{eq:parmatern_cov}
\covf\{ y_i(\bl), y_j(\bl') \} = \sigma_{ij} \gamma_{ij}  M(\bl-\bl'; (\nu_i+\nu_j)/2, \phi),
\end{equation}
where $M(\bh; \nu, \phi) =2^{1-\nu} \Gamma(\nu)^{-1} (\phi \|\bh\|)^{\nu} K_{\nu}(\phi \|\bh\|)$ is the Mat\'ern correlation function, $K_{\nu}$ is the modified Bessel function of the second kind, and \[\gamma_{ij} = \frac{\Gamma(\nu_i+\frac{d}{2})^{1/2}}{\Gamma(\nu_i)^{1/2}} \frac{\Gamma(\nu_j+\frac{d}{2})^{1/2}}{\Gamma(\nu_j)^{1/2}} \frac{\Gamma(\frac{1}{2} (\nu_i + \nu_j)) }{\Gamma(\frac{1}{2}(\nu_i + \nu_j) + \frac{d}{2})}.\] 
This parametrization of the model yields $\covf\{ \bY(\bl), \bY(\bl) \} = \bSigma \odot \bGamma_{\nu}$, where $\bGamma_{\nu}=[\gamma_{ij}]$. 
\begin{theorem}\label{thm:matern}
Under the parsimonious Mat\'ern model for $\bY(\cdot)$,
\begin{enumerate}
\item[(a)] $\text{corr}\{y_i(\bl), y_j(\bl') \} = \sigma_{ij}/(\sigma_{ii}\sigma_{jj})^{1/2} \gamma_{ij} M(\bl-\bl'; (\nu_i+\nu_j)/2, \phi)$,
\item[(b)] $\text{corr}\{y_i(\bl), y_j(\bl') \mid \sigY \} = -Q_{ij}/(Q_{ii}Q_{jj})^{1/2} \gamma_{ij} M(\bl-\bl'; (\nu_i+\nu_j)/2, \phi)$,
\item[(c)] $y_i(\cdot) \indep y_j(\cdot) \mid \bY_{\others}(\cdot)$ if and only if $Q_{ij} = 0$.
\end{enumerate}

\end{theorem}
The proof, in the Supplement, proceeds by establishing that \eqref{eq:parmatern_cov} is spectrally inside-out. Theorem~\ref{thm:matern} shows that not only does the multivariate Mat\'ern admit a nontrivial parametrization for process-level conditional independence, contrary to the conclusion in Section 3.1 of \citet{dey_graphical_2022}: it also yields a full characterization of process-level partial correlation. We have adopted the above parametrization of the model for consistency with the rest of the article; if instead we let $\covf\{\bY(\bl),\bY(\bl)\}=\bSigma=\bB\odot\bGamma_{\nu}$, partial correlations across processes are governed by $\bB^{-1}$. Finally, because the Matérn spectral density depends on $(\phi^2 + \|\bomega\|^2)^{-(\nu+d/2)}$, the more general multivariate Matérn with process-specific scale parameters $\phi_j$ cannot be written as a spectrally inside-out model; whether its parameters can encode process-level conditional independence remains an open question.

\section{Inside-out cross-covariance} \label{sec:iox}
We consider the model introduced recently in \cite{peruzzi_iox}. 
Let $\bU(\cdot) = (u_1(\cdot), \dots, u_q(\cdot))^T$ be a $q$-variate white noise Gaussian process with uncorrelated components and let $\bV(\cdot) = \bLambda \bU(\cdot)$, from which $\text{cov}\{ v_i(\bl), v_j(\bl) \} = \sigma_{ij}$. Let $\rho_j(\cdot, \cdot)$ be a user-specified univariate correlation function, for $j\in \{1,\dots, q\}$. Introduce $\calS=\{\bl_1,\dots,\bl_n\}$ as the ``reference'' locations, $\bv_j = (v_j(\bl_1), \dots, v_j(\bl_n) )^T$, and $\bL_j$ a $n \times n$ matrix such that $\bL_j\bL_j^T = \rho_j(\calS)$, e.g., the lower Cholesky factor. We build $\bY(\cdot)$ on the inside-out cross-covariance model by letting
\begin{equation}\label{eq:iox_construct}
y_j(\bl) = \bh_j(\bl)^T \bL_j \bv_j + r_j(\bl)^{1/2} v_j(\bl),\qquad  j\in\{1,\dots, q\},
\end{equation}
where we define the column vector $\bh_j(\bl) = \rho_j(\calS)^{-1} \rho_j(\calS, \bl) $ and the scalar $r_j(\bl) = \rho_j(\bl, \bl) - \bh_j(\bl)^T\rho_j(\calS, \bl)$. Let $\xi_{ij}(\bl, \bl') = 1_{\{ \bl=\bl'\} } r_i(\bl)^{1/2} r_j(\bl')^{1/2}$, where $1_{\{i=j\}}=1$ if $i=j$, zero otherwise,  and define
$f_{ij}(\bl, \bl'; \calS)$ $= \bh_i(\bl)^T \bL_i \bL_j^T \bh_j(\bl')$ $ + \xi_{ij}(\bl, \bl') $. We obtain
\begin{equation}\label{eq:iox_cov}
   \covf\{y_i(\bl), y_j(\bl')\} = \sigma_{ij} f_{ij}(\bl, \bl'; \calS), 
\end{equation}
This is a nonseparable model where cross-variable spatial dependence across $\bY(\cdot)$ follows from correlation in the underlying ``seed'' $\bV(\cdot)$ and the cross product of the respective Cholesky factors. If $q=1$, \eqref{eq:iox_construct} collapses to the modified predictive Gaussian process of \cite{modifiedpp} using the set $\calS$ as inducing points.  
Because \eqref{eq:iox_cov} is nonstationary due to its dependence on $\calS$, this model is not a spectrally inside-out model in the sense of Definition \ref{def:sio}. However, in Theorem~\ref{thm:iox} we prove that \eqref{eq:iox_cov} yields a full characterization of partial cross-correlation as well as process-level conditional independence.
\begin{theorem}\label{thm:iox}
Under the inside-out cross-covariance model for $\bY(\cdot)$, 
\begin{enumerate}
\item[(a)] $\text{corr}\{y_i(\bl), y_j(\bl') \} = \sigma_{ij}/(\sigma_{ii}\sigma_{jj})^{1/2} f_{ij}(\bl, \bl'; \calS)$,
\item[(b)] $\text{corr}\{y_i(\bl), y_j(\bl') \mid \sigY \} = -Q_{ij}/(Q_{ii}Q_{jj})^{1/2} f_{ij}(\bl, \bl'; \calS)$,
\item[(c)] $y_i(\cdot) \indep y_j(\cdot) \mid \bY_{\others}(\cdot)$  if and only if $ Q_{ij} = 0$.
\end{enumerate}

\end{theorem}
The proof is in the Supplement. 
Theorem~\ref{thm:iox} parallels Theorem~\ref{thm:sio_partcor}, with the key difference that the attenuation factor $f_{ij}(\bl, \bl'; \calS)$ is not in general a function of $\bl-\bl'$ alone, even when all marginal correlation functions $\rho_j(\cdot)$ are stationary. Nevertheless, because $f_{ij}(\bl, \bl'; \calS)$ involves no cross-process parameters, all unconditional and conditional cross-process dependence is governed by $\bSigma$ and $\bQ$, respectively. 

In practice, $\calS$ can be chosen as the set where at least one outcome is observed; this choice connects this model with the process convolution model of Section \ref{sec:convnew} and its discretized approximation \eqref{eq:discconv_S}. Let $\bY$ be the $n\times q$ matrix whose $j$th column is $\by_j = y_j(\calS)$, and $\vecop(\bY)=[ \by_1^T\ \cdots \ \by_q^T ]^T$. \eqref{eq:iox_cov} implies
\begin{equation}\label{eq:iox_cov_S}
\covf\{ \vecop(\bY) \} = \text{diag}\{\bL_1, \dots, \bL_q\} \ (\bSigma \otimes \bI_n) \ \text{diag}\{\bL_1^T, \dots, \bL_q^T\},
\end{equation}
which is structured analogously to \eqref{eq:discconv_S}. However, here we are not restricted to a low-rank approximation. In fact, the multivariate Gaussian distribution of $\vecop(\bY)$ only depends on $\bL_j^{-1}$, $j\in \{1,\dots, q\}$, so these quantities can be constructed directly using spatial conditional independence restrictions, as in Vecchia’s approximation and its process-based extensions \citep{vecchia_estimation_1988, datta_hierarchical_2016, katzfuss_general_2021, zhu_radial_2024}. These methods typically take \(\calS\) to be the set of observed locations and are particularly apt at dealing with non-smooth data exhibiting short-range spatial dependence. 

\section{Linear model of coregionalization}\label{sec:lmc}
The linear model of coregionalization is arguably the most widely adopted framework for multivariate spatial modeling, owing to its tractability even when the number of variables is not small; however, its parameters do not lend themselves to a direct conditional interpretation. To see this, define $\bW(\cdot) = (w_1(\cdot), \dots, w_q(\cdot))^T$ as a $q$-variate Gaussian process with cross-correlation $\text{cov}\{ w_i(\bl), w_j(\bl') \} = 1_{\{i=j\}}\rho_j(\bl, \bl')$, where the $\rho_j(\bl, \bl')$ are stationary correlation functions with spectral density $\tilde{\rho}_j(\bomega)>0$ for almost all $\bomega$ and all $j$. We assume full nonseparability, i.e. $\rho_i \neq \rho_j$ for all $i\neq j$.
Setting $\bY(\cdot) = \bLambda \bW(\cdot)$ yields 
\begin{equation}\label{eq:lmc}
\begin{aligned}
y_j(\bl) &= \sum_{r=1}^q \lambda_{jr} w_r(\bl), \qquad \text{cov}\{ y_i(\bl), y_j(\bl') \} = \sum_{r=1}^q \lambda_{jr}\lambda_{ir} \rho_r(\bl, \bl').
\end{aligned}
\end{equation}
\begin{theorem}\label{thm:lmc}
Let $\bA=[a_{ij}]=\bLambda^{-1}$. Under a linear model of coregionalization,
\[y_i(\cdot) \indep y_j(\cdot) \mid \bY_{\others}(\cdot) \text{ if and only if } a_{ri}a_{rj}=0 \text{ for all } r=1,\dots,q.\]

\end{theorem}
The proof is in the Supplement. 
The main take-away of Theorem~\ref{thm:lmc} is that in this model, $[\bA^T\bA]_{ij}=Q_{ij}=0$ is not sufficient for process-level conditional independence; this contrasts with all the models we have considered above. 
Furthermore, we show in the Supplement that the spectral density of the residual process takes the form $\bS_Z(\bomega)=\bL_x(\bomega) \bP_o^{\perp}(\bomega)\bL_x(\bomega)^T$, where both $\bL_x(\bomega)$ and the projection matrix $\bP_o^{\perp}(\bomega)$ depend on $\bomega$ through $\bLambda$ and each $\tilde{\rho}_j(\bomega)$. This  prevents the partial cross-correlation function from factorizing into graph-level and spatial components as in Theorem~\ref{thm:sio_partcor}. 
Therefore, the linear model of coregionalization is limited to analyses of presence or absence of process-level conditional independence through the necessary and sufficient condition that $a_{ri}a_{rj}=0$ for all $r$, which translates to the columns of $\bA$ having disjoint support. This condition is related to the conditional specification of the model \citep{schmidtgelfand, cressie_multivariate_2016}, which lets $\bY(\cdot) = \bB\bY(\cdot) + \bW(\cdot)$ where $\bW(\cdot)$ is as above, $\bB$ is a strictly lower-triangular $q\times q$ matrix, and then $\bLambda^{-1} = \bA = \bI_q - \bB$. In this triangular system of $q$ linear regressions, sparsity on $\bA$ is induced by sparsity on $\bB$, defining a directed acyclic graph which, after moralization, leads to undirected graphical dependence. Order dependence and other issues with this specification were mentioned in \cite{gelfand_nonstationary_2004}. 
Our analysis here has focused on the full-rank case, but the primary appeal of linear coregionalization is that we can choose $\bLambda$ to be ``tall and skinny;'' this is convenient as a prior for latent effects, yielding parsimony and reduced complexity through dimension reduction. 
However, in that case $\bLambda^{-1}$ does not exist, and we cannot apply Theorem \ref{thm:lmc}. Conditional independence may still be enforced via restrictions on the null space of $\bLambda$ and entries of its Moore-Penrose pseudoinverse, but enforcing such conditions negates the tractability that motivates this model.

\section{Illustrations} \label{sec:illu}
We illustrate the novel inferences afforded by inside-out models relative to traditional spatial multivariate analyses in three ways. First, we generate spatially correlated data from a parsimonious Mat\'ern model and highlight the inferential gains from estimating presence or absence of graph edges across processes rather than assuming independence. Second, we estimate the full cross-correlation functions from simulated data from the model in Figure \ref{fig:partial_corr}. Third, we consider the Jura dataset, a well-studied benchmark in multivariate spatial statistics with measurements of seven heavy metals at irregularly spaced locations in the Swiss Jura mountains.
Our comparisons involve different models as well as different algorithms and software tools. We fit parsimonious Mat\'ern models via \texttt{GpGpm} \citep{fahmy_vecchia_2022}, the inside-out cross-covariance model via \texttt{spiox}, the linear model of coregionalization via \texttt{meshed}. All software is publicly available on Github. Source code to reproduce all analyses is part of the supplementary material. 

\subsection{Process-level graph recovery}
\begin{figure}[!htb]
\centering
\includegraphics[width=0.95\textwidth]{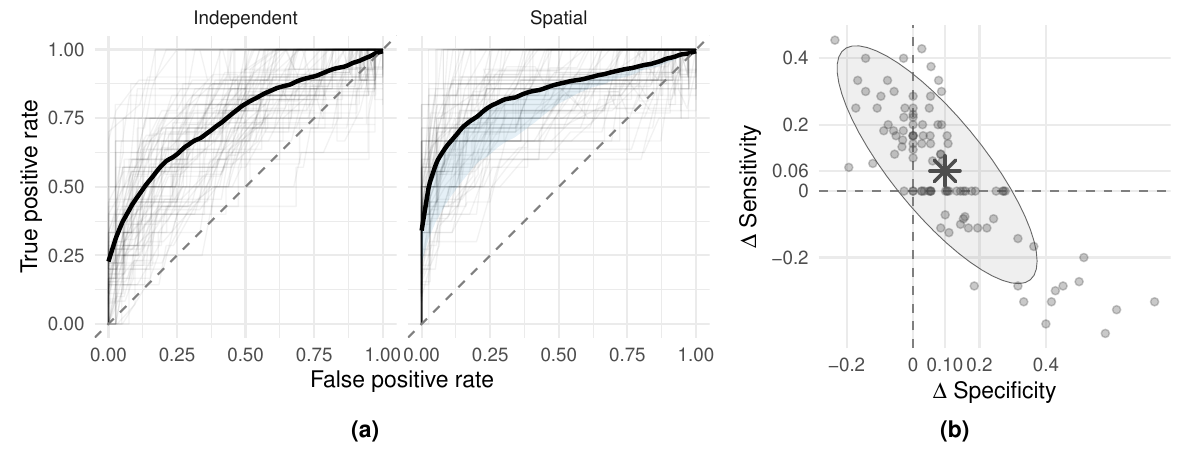}
\caption{Graph recovery via graphical lasso applied to $q=10$ processes observed at $n=1024$ locations, across $100$ replicate datasets. The \textit{Spatial} method uses the covariance estimate from a fitted parsimonious multivariate Mat\'ern; the \textit{Independent} method uses the sample covariance treating observations as independent replicates. (a)~ROC curves for each replicate (light lines), pointwise averages (bold), and mean gain of \textit{Spatial} over \textit{Independent} (shaded). (b)~Paired differences in sensitivity and specificity at each replicate's best $F_1 = \frac{2 \cdot \text{precision} \cdot \text{recall}}{\text{precision} + \text{recall}}$ threshold; the asterisk marks the mean, and positive values indicate better performance of the \textit{Spatial} model.}
\label{fig:illu1_glasso}
\end{figure}

Ignoring spatial dependence may lead to incorrect inferences on conditional independence. We assess graph recovery in a simulation study based on the parsimonious multivariate Mat\'ern model. For each of $100$ replicate datasets, we sample a random graph on $q=10$ nodes using the \texttt{BDgraph} package \citep{BDgraph}, with edge probability $0.2$, then draw a precision matrix from the G-Wishart distribution with $5$ degrees of freedom and identity scale matrix. For numerical stability across replicates, we rescale to obtain $\bQ$ with unit diagonal, and compute $\bSigma = \bQ^{-1}$. We generate parsimonious multivariate Mat\'ern data \eqref{eq:parmatern_cov} on a regular grid of $n=1024$ locations on $\calD = [0,1]^2$, drawing $\nu_j\iidsim U[0.5, 2]$ for $j=1,\dots, q$ and fixing $\phi=10$.

To isolate the effect of modeling spatial dependence, we compare two approaches to graph recovery via the graphical lasso \citep{friedman_sparse_2008}. The first, labeled \textit{Spatial}, fits the parsimonious multivariate Mat\'ern  with $\nu_j$ and $\phi$ fixed at their simulated values, transforms the output to obtain $\hat{\bSigma}$ matching our parametrization, and applies the graphical lasso to this estimate. The second, labeled \textit{Independent}, ignores spatial dependence entirely and applies the graphical lasso to the sample covariance matrix computed from the $n$ observations treated as independent replicates. Figure~\ref{fig:illu1_glasso}(a) shows the ROC curves for all $100$ replicates along with their pointwise averages; the shaded region in the \textit{Spatial} panel highlights the average gain relative to the \textit{Independent} approach. The \textit{Spatial} method achieves a mean area under the ROC curve of $\overline{\text{AUC}} = 82.7\%$ compared to $73.1\%$ for \textit{Independent}. Panel~(b) displays paired differences in sensitivity and specificity at each replicate's optimal $F_1$ threshold. The mean difference (asterisk) is $(9.6\%, 6.1\%)$, indicating that accounting for spatial dependence yields improvements in both specificity and sensitivity on average: across replicates, the \textit{Spatial} approach improves or matches both sensitivity and specificity in 53\% of cases, and improves or matches at least one on all datasets.

\subsection{Analysis of synthetic datasets} \label{sec:simulations}
We target estimation of the full set of partial correlation functions. We simulate a $q=5$ variate Gaussian process on a $25 \times 25$ regular grid over $[0,1]^2$, giving $n=625$ observations per outcome. The process follows the parsimonious multivariate Mat\'ern of Figure \ref{fig:partial_corr}. 
All model parameters are estimated from data. Fitting times were 10.4 and 19.6 seconds for the parsimonious Mat\'ern and the inside-out cross-covariance, respectively. Figure~\ref{fig:illu2_pmatern} displays the estimated unconditional and partial spatial cross-correlation functions for all pairs involving $y_3(\cdot)$, along with the true functions in dotted lines. Both models broadly agree on both unconditional and partial cross-correlation functions, with the correctly specified parsimonious Mat\'ern showing more accurate estimates on this dataset. In particular, both models correctly recover the signs of all nonzero partial correlation coefficients, including the unequal sign between the partial and the unconditional cross-correlation functions between $y_2(\cdot)$ and $y_3(\cdot)$. 
A traditional spatial analysis of this dataset would be limited to marginal correlation and cross-correlation functions; in this particular setting such an analysis would have missed the positive conditional association between $y_2(\cdot)$ and $y_3(\cdot)$, as well as the lack of conditional dependence between $y_3(\cdot)$ and $y_5(\cdot)$. Our analysis of the partial correlation functions does not depend on any specific prior or model fitting choice: all results were obtained via the default fitting methods provided in the available software packages.

\begin{table}
\caption{\label{tab:illu2_summary} Estimation, prediction, and uncertainty quantification performance averaged over 200 simulated datasets. IOX: inside-out cross-covariance; LMC: linear model of coregionalization. The LMC does not yield process-level partial correlations, hence the empty cells. CRPS is unavailable for the parsimonious Mat\'ern because \texttt{GpGpm} returns point estimates.}
\fbox{%
\begin{tabular}[t]{lrrrrrr}
\toprule
\multicolumn{1}{c}{ } & \multicolumn{3}{c}{RMSE} & \multicolumn{3}{c}{CRPS} \\
\cmidrule(l{3pt}r{3pt}){2-4} \cmidrule(l{3pt}r{3pt}){5-7}
 & IOX & P. Matérn & LMC & IOX & P. Matérn & LMC\\
\midrule
Partial correlations & 0.104 & 0.058 &  & 0.056 &  & \\
Marginal variance & 0.300 & 0.217 & 1.383 & 0.170 &  & 0.451\\
Cross-covariance & 0.167 & 0.103 & 0.788 & 0.095 &  & 0.242\\
Predictions & 0.743 & 1.155 & 0.776 & 0.414 & 0.651 & 0.433\\
\bottomrule
\end{tabular}}
\end{table}

\begin{figure}[!htb]
\centering
\includegraphics[width=0.95\textwidth]{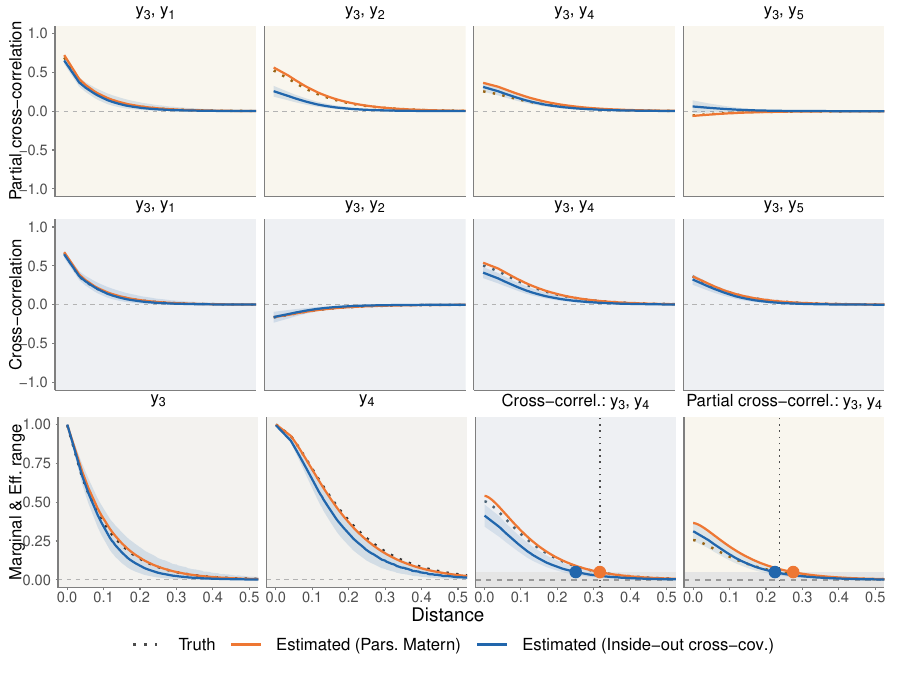}
\caption{Estimating unconditional and conditional spatial dependence on parsimonious Mat\'ern data. \textit{Top row}: true and estimated partial cross-correlations between $y_3(\cdot)$ and other processes. \textit{Middle row}: true and estimated cross-correlations between $y_3(\cdot)$ and other processes. \textit{Bottom row, left}: true and estimated marginal correlations of the processes $y_3(\cdot)$ and $y_4(\cdot)$. \textit{Bottom row, right}: estimated effective unconditional and conditional cross-ranges between $y_3(\cdot)$ and $y_4(\cdot)$. The vertical dotted lines correspond to the true effective ranges, the circles are colored according to the estimated model. Shaded blue areas correspond to the 95\% credible intervals for the inside-out cross-covariance model estimated via Markov-chain Monte Carlo. }
\label{fig:illu2_pmatern}
\end{figure}

We validate our findings across different scenarios by generating $200$ independent datasets as above; we maintain the same graphical structure, but vary smoothness and range parameters by drawing them, respectively, as $\nu_j \stackrel{\text{iid}}{\sim} \text{Uniform}(0.2, 1.8)$ for $j=1,\dots,q$ and $\phi \sim \text{Uniform}(10, 40)$. For analysing predictive performance, we remove $200$ entries of the $n \times q$ response matrix uniformly at random to create a misaligned dataset. In addition to the parsimonious Mat\'ern and the inside-out cross-covariance, we fit a linear model of coregionalization. We compare models via root mean squared error and continuous ranked probability score (CRPS) in estimating marginal variance and cross-covariance at distance zero (the diagonal and off-diagonal entries of the coregionalization matrix $\bSigma \odot \Gamma_{\nu}$, respectively), the partial correlations at zero distance $-{Q_{ij}}/(Q_{ii}Q_{jj})^{1/2}\gamma_{ij}$, and root mean squared prediction error and CRPS at the held-out locations. We report average performance across all datasets and all processes in Table \ref{tab:illu2_summary}, which shows that the parsimonious Mat\'ern model is more accurate in estimating model parameters, but slightly worse in predictive performance. Both the parsimonious Mat\'ern and the inside-out cross-covariance accurately recover partial correlations; the inside-out model also yields better uncertainty quantification, as evidenced by the lower CRPS scores. The average compute times were 21, 65, 20 seconds respectively for the parsimonious Mat\'ern, inside-out cross-covariance, linear model of coregionalization.

\subsection{Jura data}
The Jura dataset comprises 359 irregularly spaced topsoil measurements of seven heavy metals (cadmium, cobalt, chromium, copper, nickel, lead, zinc) collected across a 14.5$\times$2.5 km region in the Swiss Jura mountains. The data are included in the \texttt{gstat} package \citep{gstat}; additional details are available in \cite{goovaerts_geostatistics_1997}. 
We work with log-transformed, centered concentrations. We create a test set using 200 uniformly sampled measurements from the $nq=$ 2513 total. We fit multivariate Gaussian processes equipped with the parsimonious Mat\'ern, the inside-out cross-covariance, and the linear model of coregionalization. We split our analysis into two parts. The first is a standard geostatistical analysis, where we estimate covariance parameters and cross-correlation functions, and compare models in terms of uncertainty quantification and predictive performance on the test set. The second targets process-level partial correlations and the implied conditional independence graph. The split is purely expository: for both inside-out models, the partial correlation structure is a byproduct of estimation and requires no additional modeling effort.

\begin{figure}[!htb]
\centering
\includegraphics[width=0.95\textwidth]{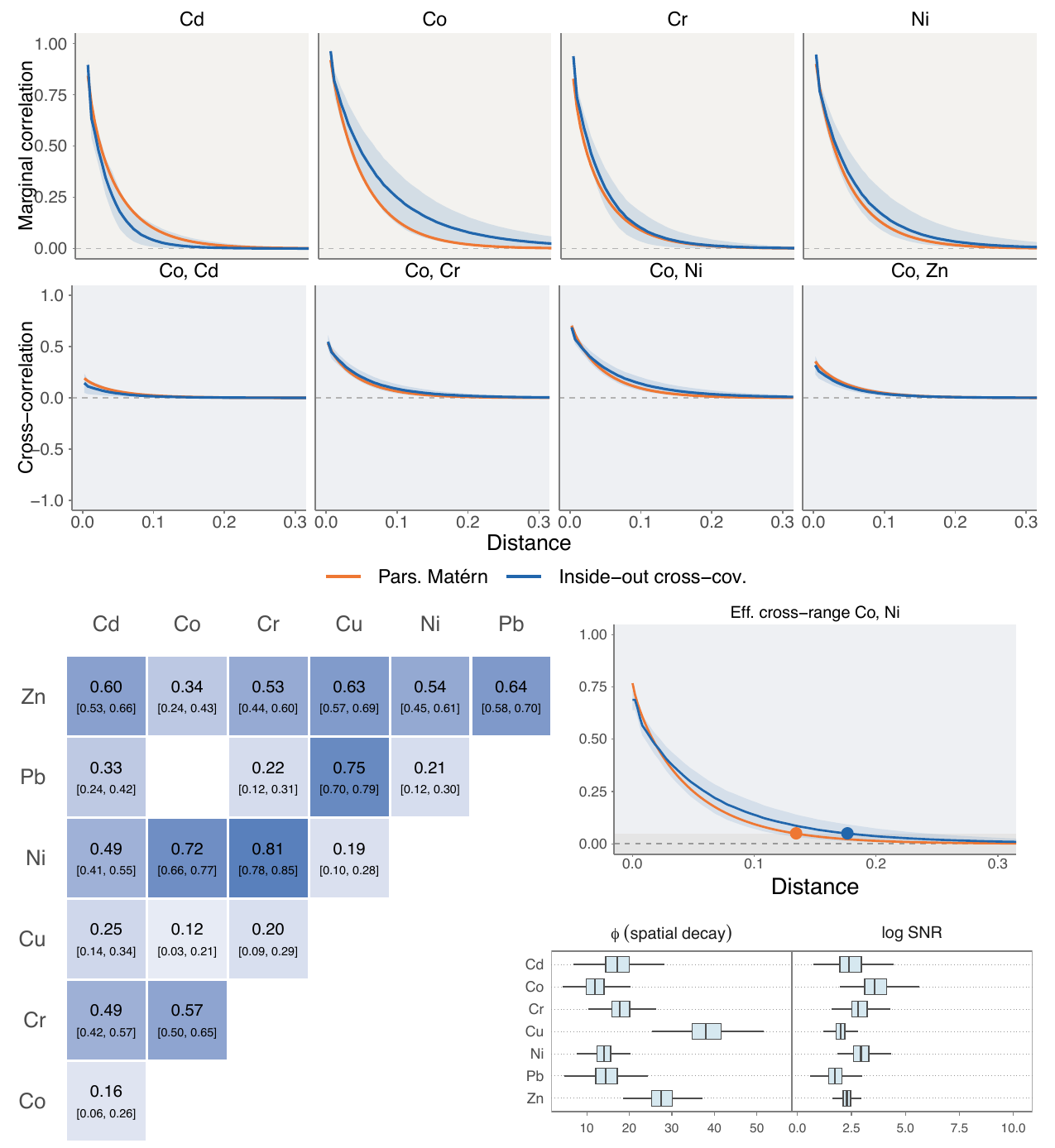}
\caption{Geostatistical analysis of the Jura dataset via parsimonious Mat\'ern and inside-out cross-covariance models. \textit{First row}: marginal correlation functions for four of the seven heavy metals. \textit{Second row}: cross-correlation functions between cobalt and four other heavy metals. \textit{Bottom left}: cross-correlation at zero spatial distance, with 95\% credible intervals. \textit{Bottom right}: effective cross-range between cobalt and nickel, and posterior summaries of the spatial decay parameter and the log signal-to-noise ratio from the inside-out cross covariance model.}
\label{fig:jura_traditional}
\end{figure}

\begin{table}
\caption{\label{tab:jura_summary} Predictive performance in terms of root mean squared error (RMSPE) and CRPS across the different heavy metals and methods. The last row refers to average performance across all metals.}
\fbox{%
\begin{tabular}[t]{lrrrrrr}
\toprule
\multicolumn{1}{c}{ } & \multicolumn{3}{c}{RMSPE} & \multicolumn{3}{c}{CRPS} \\
\cmidrule(l{3pt}r{3pt}){2-4} \cmidrule(l{3pt}r{3pt}){5-7}
Outcome & IOX & P. Matérn & LMC & IOX & P. Matérn & LMC\\
\midrule
Cd & 0.441 & 0.557 & 0.462 & 0.251 & 0.310 & 0.259\\
Co & 0.205 & 0.244 & 0.218 & 0.116 & 0.132 & 0.125\\
Cr & 0.160 & 0.220 & 0.170 & 0.092 & 0.129 & 0.096\\
Cu & 0.324 & 0.568 & 0.266 & 0.193 & 0.320 & 0.159\\
Ni & 0.159 & 0.264 & 0.161 & 0.090 & 0.148 & 0.093\\
Pb & 0.245 & 0.318 & 0.255 & 0.133 & 0.188 & 0.145\\
Zn & 0.160 & 0.284 & 0.172 & 0.085 & 0.146 & 0.097\\
\addlinespace
Overall & 0.263 & 0.380 & 0.263 & 0.138 & 0.198 & 0.140\\

\bottomrule
\end{tabular}}
\end{table} 

\subsubsection{Geostatistical analysis and comparisons}
A traditional geostatistical analysis of the data reveals that heavy metals exhibit different spatial ranges and signal-to-noise ratios. Copper (\textit{Cu}) exhibits the highest spatial decay, indicating the shortest range of spatial dependence, while cobalt (\textit{Co}) shows the highest signal-to-noise ratio. All metals co-vary positively across the region, as evidenced by the positive co-located correlation coefficients (Figure \ref{fig:jura_traditional}, bottom left); the strongest correlations are observed in the nickel (\textit{Ni}) and chromium (\textit{Cr}) pair, and the lead (\textit{Pb}) and copper pair, respectively 0.81 and 0.75. Furthermore, we observe that cobalt and nickel exhibit long range spatial co-variability. Comparing modeling approaches, uncertainty bands on correlation functions from the inside-out cross-covariance model generally include the point estimates from the parsimonious Mat\'ern model, from which we conclude that the two models broadly agree on this dataset. Table~\ref{tab:jura_summary} validates our findings by comparing the results of our model fitting with a linear model of coregionalization. We observe that the predictive performance of the parsimonious Mat\'ern model is generally worse than the other two, and that overall, the inside-out cross-covariance model outperforms others.

\subsubsection{Partial correlation network analysis}

\begin{figure}[!htb]
\centering
\includegraphics[width=0.95\textwidth]{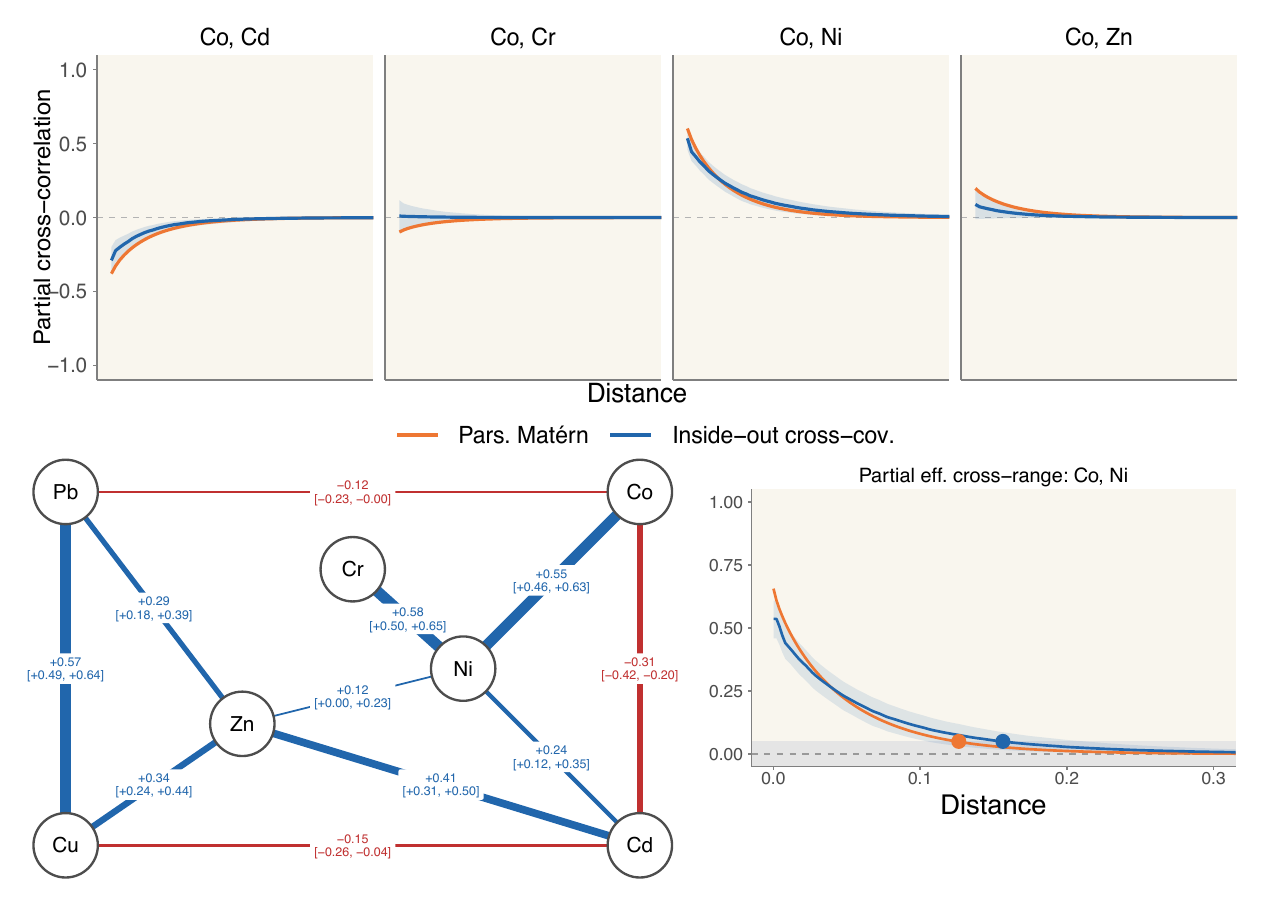}
\caption{Process-level partial correlation analysis of the Jura dataset via parsimonious Mat\'ern and inside-out cross covariance models. \textit{Top}: partial cross-correlation functions for four of the seven heavy metals. \textit{Bottom left}: partial correlation network, with red and blue colors indicating negative and positive conditional associations, respectively, and edge thickness corresponding to the magnitude of such conditional dependence. The labels and corresponding uncertainty bands refer to the partial correlation coefficients $-Q_{ij}/(Q_{ii}Q_{jj})^{1/2}$. \textit{Bottom right}: partial effective cross-range between cobalt and nickel.}
\label{fig:jura_partial}
\end{figure}

The parsimonious Mat\'ern and the inside-out cross-covariance models yield spatial partial correlations and a conditional independence graph at no additional modeling cost. This analysis reveals a much more nuanced view of spatial co-variability of the heavy metals in this region. First, some positive associations that we found above are confirmed, but weaker: specifically, between nickel and cobalt, nickel and chromium, and copper and lead. Second, several edges are missing from the graphical model, indicating pervasive conditional independence structure. For example, we find that spatial dependence between chromium and cobalt is entirely mediated by nickel, as the two processes are estimated to be conditionally independent; however, classical geostatistical analysis had revealed strong and positive spatial association. Furthermore, the (mild) positive spatial association between cadmium and cobalt is revealed to actually be a negative association after accounting for other metals: from the estimated graphical model (Figure~\ref{fig:jura_partial}, bottom left), we infer that the observed positive unconditional dependence is due to positive association of both cobalt and cadmium with nickel. A similar sign reversal occurs between copper and cadmium, with the positive unconditional association attributable to both metals' positive partial association with zinc.
Our findings show that unconditional cross-correlation functions can obscure or reverse the sign and strength of direct process-level associations, and underscore the value of the partial correlation perspective enabled by inside-out models.

\section{Discussion}

We have shown that a broad class of cross-covariance models enable a rich analysis of partial correlation at the process level, which can reveal a more nuanced view of spatial dependence through a graphical interpretation. These models fill a gap in traditional analyses of multivariate spatial data, which typically do not provide graphical interpretations. Our results establish a direct link between multivariate spatial modeling and covariance selection \citep{dempster_covariance_1972}, penalized network estimation \citep{friedman_sparse_2008, whittaker_graphical_2009, danaher_joint_2014}, and Bayesian graph learning \citep{wang_scaling_2015, vogels_bayesian_2024}. More broadly, the correspondence between $\bQ$ and process-level network structure suggests natural connections to relational and latent structure modeling \citep{holland_stochastic_1983, hoff_latent_2002}, where the goal is to characterize the generative mechanism that produces the graphical model.

Our cross-covariance constructions extend naturally to hierarchical latent-Gaussian models, where $\bQ$ governs the conditional dependence structure of the latent processes regardless of the observation model. Ecological applications are particularly compelling, since graphical structure yields direct interpretations of interactions among species populations \citep{warton_so_2015, ovaskainen_joint_2020, wang_role_2023b}. In these contexts, 
allowing partial correlations to vary spatially or depend on covariates would enable the latent conditional independence structure to adapt to changing environmental or experimental conditions; we anticipate these as interesting directions for future research.  

Finally, distinguishing direct from mediated relationships is most valuable in high-dimensional settings where only a few direct associations are expected. In such settings, methods that support both scalable computation and interpretable conditional dependence stand to yield substantial gains in the analysis of large multivariate spatial datasets. Spatial factor models derived from linear coregionalization achieve scalability via dimension reduction, but cannot encode meaningfully conditional dependence. More promising directions may lie in constructions that operate directly on the precision structure (e.g., \citealt{chandra_bayesian_2026}), offering a natural basis for scalable graphical modeling of multivariate spatial processes.  


\spacingset{1}
{\small
\bibliographystyle{Arxiv/agsm}
\bibliography{arxiv_references}
}

\newpage

\spacingset{1.4} 
\begin{center}
{\Large\bf Supplementary material}
\end{center}

\appendix
\setcounter{theorem}{0}
\section{Proofs}
\begin{theorem}\label{thm:sio_partcor_supp}

If $\bY(\cdot)$ is spectrally inside-out, then $\text{cov}\{y_i(\bl), y_j(\bl')\}=\sigma_{ij} c_{ij}(\bl-\bl')$, where $c_{ij}(\cdot)$ does not depend on $\bSigma$ or $\bQ$, and 
\begin{enumerate}
\item[(a)] $\text{corr}\{y_i(\bl), y_j(\bl') \} = \sigma_{ij}/(\sigma_{ii}\sigma_{jj})^{1/2} c_{ij}(\bl - \bl')$,
\item[(b)] $\text{corr}\{y_i(\bl), y_j(\bl') \mid \sigY \} = -Q_{ij}/(Q_{ii}Q_{jj})^{1/2} c_{ij}(\bl - \bl'),$
\item[(c)] $y_i(\cdot) \indep y_j(\cdot) \mid \bY_{\others}(\cdot)$ if and only if $Q_{ij} = 0$.
\end{enumerate}

\end{theorem}
\begin{proof}
For the spatial correlation function, taking inverse Fourier transforms,
$ \text{cov}\{  y_i(\bl), y_j(\bl')\} = \intRd e^{\mathrm{i}\bomega\cdot (\bl-\bl')} [\bS_Y(\bomega)]_{ij} d\bomega = \sigma_{ij} c_{ij}(\bl-\bl')$, where $c_{ij}(\bl-\bl')=\intRd e^{\mathrm{i}\bomega\cdot (\bl-\bl')} d_i(\bomega)d_j(\bomega)^* d\bomega$ does not depend on $\bSigma$, 
$\text{var}\{ y_i(\bl) \}=\text{cov}\{  y_i(\bl), y_i(\bl) \}=\sigma_{ii}$ since $\intRd |d_i(\bomega)|^2 d\bomega = 1$ by assumption, and similarly $\text{var}\{ y_j(\bl') \}= \sigma_{jj}$. 
We find the correlation function 
\[ \text{corr}\{y_i(\bl), y_j(\bl')\} = \frac{\text{cov}\{  y_i(\bl), y_j(\bl')\}}{(\text{var}\{ y_i(\bl)\}\text{var}\{ y_j(\bl') \})^{1/2}} = \frac{\sigma_{ij}}{(\sigma_{ii}\sigma_{jj})^{1/2}} c_{ij}(\bl-\bl'). \]

For the partial cross-correlation function, let $\bY_x(\cdot) = (y_i(\cdot), y_j(\cdot))^T$, so that $\bY(\cdot) = (\bY_x(\cdot)^T, \bY_o(\cdot)^T)^T$. The spectral density matrix of $\bZ(\cdot) = (z_i(\cdot), z_j(\cdot))^T$ is the $2\times 2$ matrix $\bS_Z(\bomega) = \bS_x(\bomega) - \bS_{x,o}(\bomega)\, \bS_{o}(\bomega)^{-1}\, \bS_{o,x}(\bomega)$, where $\bS_{x,o}(\bomega)$ is the submatrix of $\bS_Y(\bomega)$ choosing rows $\{i,j\}$ and columns $o = \{1, \dots, q \}\setminus \{ i,j\}$, $\bS_o(\bomega)$ subsets to $o$ rows and columns. Similarly for $\bSigma_x, \bSigma_{x,o}, \bSigma_{o}, \bSigma_{o,x}$.  
Since $\bD(\bomega)$ is diagonal, we obtain
\begin{align*}
\bS_Z(\bomega) 
&= \bD_x(\bomega)\bSigma_x\bD_x(\bomega)^* \\
&\qquad\qquad - \bD_{x}(\bomega)\bSigma_{x,o}\bD_{o}(\bomega)^* (\bD_o(\bomega)^*)^{-1} \bSigma_o^{-1} \bD_o(\bomega)^{-1} \bD_{o}(\bomega)\bSigma_{o,x}\bD_{x}(\bomega)^*\\
&= \bD_x(\bomega)\bSigma_x\bD_x(\bomega)^* - \bD_{x}(\bomega)\bSigma_{x,o} \bSigma_{o}^{-1} \bSigma_{o,x}\bD_{x}(\bomega)^*\\
&= \bD_x(\bomega) (\bSigma_x - \bSigma_{x,o}\bSigma_{o}^{-1}\bSigma_{o,x}) \bD_x(\bomega)^*\\
&= \bD_x(\bomega) \bQ_x^{-1} \bD_x(\bomega)^*,
\end{align*}
where, by the properties of block matrix inverses, $\bQ_x^{-1}$ is the inverse of the $x$ block of $\bQ$.
The off-diagonal element of $\bQ_{x}^{-1}$ is $-Q_{ij}/(Q_{ii}Q_{jj} - Q_{ij}^2)$, therefore
\begin{align*}
&[\bS_Z(\bomega)]_{1,1} = \frac{Q_{jj}}{Q_{ii}Q_{jj} - Q_{ij}^2}|d_i(\bomega)|^2,\quad  \quad 
[\bS_Z(\bomega)]_{2,2} = \frac{Q_{ii}}{Q_{ii}Q_{jj} - Q_{ij}^2}|d_j(\bomega)|^2,\\
&[\bS_Z(\bomega)]_{1,2} = \frac{-Q_{ij}}{Q_{ii}Q_{jj} - Q_{ij}^2}d_i(\bomega) d_j(\bomega)^*.
\end{align*}
Taking inverse Fourier transforms of each of the above, we find
\begin{align*}
&\text{var}\{z_i(\bl)\} 
= \frac{Q_{jj}}{Q_{ii}Q_{jj} - Q_{ij}^2}, \quad
\text{var}\{z_j(\bl)\} 
= \frac{Q_{ii}}{Q_{ii}Q_{jj} - Q_{ij}^2} \\
&\covf\{z_i(\bl),\, z_j(\bl')\} 
= \frac{-Q_{ij}}{Q_{ii}Q_{jj} - Q_{ij}^2} 
c_{ij}(\bl-\bl').
\end{align*}
The partial cross-correlation function is then defined as the spatial cross-correlation function of the residual processes, which we find
\begin{align*}
\text{corr}\{y_i(\bl), y_j(\bl')\mid \Yothers \} &= \text{corr}\{z_i(\bl),\, z_j(\bl')\}= \frac{\covf\{z_i(\bl),\, z_j(\bl')\}}{(\text{var}\{z_i(\bl)\} \text{var}\{z_j(\bl')\} )^{1/2}} \\
&= \frac{-Q_{ij}}{(Q_{ii} Q_{jj})^{1/2}}\; c_{ij}(\bl - \bl').
\end{align*}
Finally, the conditional independence part follows since $\text{corr}\{z_i(\bl),\, z_j(\bl')\}=0$ if and only if $Q_{ij}=0$ above. We can alternatively see this also by applying Theorem 2.4 of \cite{dahlhaus_graphical_2000}. Since we have
$\bS_Y(\bomega)^{-1} = (\bD(\bomega)^*)^{-1}\, \bQ\, \bD(\bomega)^{-1}$ for almost all $\bomega$, then $[\bS_Y(\bomega)^{-1}]_{ij} = Q_{ij}/\{d_i(\bomega)^*d_j(\bomega)\}$. Since $d_i(\bomega)$ and $d_j(\bomega)$ are assumed nonzero for almost all $\bomega$, $[\bS_Y(\bomega)^{-1}]_{ij}=0$ if and only if $Q_{ij}=0$, which proves conditional independence at the process level.
\end{proof}

\begin{theorem}\label{thm:intrinsic_supp}
Under the intrinsic (separable) coregionalization model for $\bY(\cdot)$,
\begin{enumerate}
\item[(a)] $\text{corr}\{y_i(\bl), y_j(\bl') \} = \sigma_{ij}/(\sigma_{ii}\sigma_{jj})^{1/2} \rho(\bl - \bl')$,
\item[(b)] $\text{corr}\{y_i(\bl), y_j(\bl') \mid \sigY \} = -Q_{ij}/(Q_{ii}Q_{jj})^{1/2} \rho(\bl - \bl')$,
\item[(c)] $y_i(\cdot) \indep y_j(\cdot) \mid \bY_{\others}(\cdot)$ if and only if $ Q_{ij} = 0$
\end{enumerate}
\end{theorem}
\begin{proof}
We can write $\bY(\bl)=\bLambda \bW(\bl)$ where $\bS_W(\bomega) = \tilde{\rho}(\bomega) \bI_q$, thus $\bS_Y(\bomega) = \{ \tilde{\rho}(\bomega)^{1/2} \bI_q \} \bSigma \{ \tilde{\rho}(\bomega)^{1/2} \bI_q \}$ and $\bS_Y(\bomega)^{-1} = \{ \tilde{\rho}(\bomega)^{-1/2} \bI_q \} \bQ \{ \tilde{\rho}(\bomega)^{-1/2} \bI_q \}$. Because $Y(\cdot)$ can be cast as a spectrally inside-out process, we can use Theorem~\ref{thm:sio_partcor_supp}, where we then compute $c_{ij}(\bl-\bl')=\intRd e^{\mathrm{i}\bomega\cdot (\bl-\bl')}\tilde{\rho}(\bomega)d\bomega=\rho(\bl-\bl')$.
\end{proof}

\begin{theorem}\label{thm:convnew_supp}
Under the component-specific process convolution model for $\bY(\cdot)$, 
\begin{enumerate}
\item[(a)] $\text{corr}\{y_i(\bl), y_j(\bl') \} = \sigma_{ij}/(\sigma_{ii}\sigma_{jj})^{1/2} k_{ij}(\bl-\bl')$,
\item[(b)] $\text{corr}\{y_i(\bl), y_j(\bl') \mid \sigY \} = -Q_{ij}/(Q_{ii}Q_{jj})^{1/2} k_{ij}(\bl-\bl')$,
\item[(c)] $y_i(\cdot) \indep y_j(\cdot) \mid \bY_{\others}(\cdot)$ if and only if $ Q_{ij} = 0$.
\end{enumerate}
\end{theorem}
\begin{proof}
We show that the component-specific convolution model is spectrally inside-out.
Let $\tilde{y}_j(\bomega)$, and $\tilde{v}_j(\bomega)$ be the Fourier transforms of $y_j(\cdot)$, and $v_j(\cdot)$, respectively, for $j \in \{ 1, \dots, q\}$. By the convolution theorem, $\tilde{y}_j(\bomega) = d_j(\bomega) \tilde{v}_j(\bomega)$. Let $\bK(\bl) = \text{diag}\{k_1(\bl), \dots, k_q(\bl)\}$ and similarly $\bD(\bomega)=\text{diag}\{d_1(\bomega), \dots, d_q(\bomega) \}$. Then $\bY(\bl) = (\bK \star \bV)(\bl)$ and $\tilde{\bY}(\bomega) = \bD(\bomega) \tilde{\bV}(\bomega)$. Because $\bV(\cdot)$ is spatially uncorrelated and $\bLambda \bLambda^T=\bSigma$, we have $\bS_V(\bomega) \propto \bSigma$, therefore $\bS_Y(\bomega) = \bD(\bomega) \bS_V(\bomega) \bD(\bomega)^* \propto  \bD(\bomega) \bSigma \bD(\bomega)^*$. Because this is a spectrally inside-out model, we can apply Theorem \ref{thm:sio_partcor_supp}, where we find $c_{ij}(\bl-\bl') = k_{ij}(\bl - \bl')$.
\end{proof}

\begin{theorem}\label{thm:matern_supp}
Under the parsimonious Mat\'ern model for $\bY(\cdot)$,
\begin{enumerate}
\item[(a)] $\text{corr}\{y_i(\bl), y_j(\bl') \} = \sigma_{ij}/(\sigma_{ii}\sigma_{jj})^{1/2} \gamma_{ij} M(\bl-\bl'; (\nu_i+\nu_j)/2, \phi)$,
\item[(b)] $\text{corr}\{y_i(\bl), y_j(\bl') \mid \sigY \} = -Q_{ij}/(Q_{ii}Q_{jj})^{1/2} \gamma_{ij} M(\bl-\bl'; (\nu_i+\nu_j)/2, \phi)$,
\item[(c)] $y_i(\cdot) \indep y_j(\cdot) \mid \bY_{\others}(\cdot)$ if and only if $Q_{ij} = 0$.
\end{enumerate}

\end{theorem}
\begin{proof}
We prove that the parsimonious Mat\'ern is spectrally inside-out. The conclusion will then follow from Theorem~\ref{thm:sio_partcor_supp}. 
Let $m(\bomega; \nu, \phi)$ be the spectral density of a Mat\'ern correlation function with smoothness $\nu$ and scale $\phi$ in $\RRd$:  
\[m(\bomega; \nu, \phi) = \frac{\Gamma(\nu + \tfrac{d}{2}) \phi^{2\nu}}{\Gamma(\nu) \pi^{d/2}} (\phi^2 + \|\bomega\|^2)^{-(\nu + d/2)}.\]
The parsimonious Mat\'ern sets $\nu_{ij}=(\nu_i+\nu_j)/2$ and $\phi_{ij}=\phi$ for all $i,j=1,\dots,q$. Therefore, the $i,j$ element of the spectral density matrix of $\bY(\cdot)$ is
\begin{align*}
[\bS_Y(\bomega)]_{ij} &= \sigma_{ij} \gamma_{ij} m(\bomega; (\nu_i+\nu_j)/2, \phi)\\
&=\sigma_{ij} \gamma_{ij}\frac{\Gamma(\tfrac{1}{2}(\nu_i+\nu_j) + \tfrac{d}{2}) \phi^{\nu_i+\nu_j}}{\Gamma(\tfrac{1}{2}(\nu_i+\nu_j)) \pi^{d/2}} (\phi^2 + \|\bomega\|^2)^{-(\frac{\nu_i+\nu_j}{2} + \frac{d}{2})}.
\end{align*}
Substituting $\gamma_{ij} = \frac{\Gamma(\nu_i+\frac{d}{2})^{1/2}}{\Gamma(\nu_i)^{1/2}} \frac{\Gamma(\nu_j+\frac{d}{2})^{1/2}}{\Gamma(\nu_j)^{1/2}} \frac{\Gamma(\frac{1}{2}(\nu_i+\nu_j))}{\Gamma(\frac{1}{2}(\nu_i+\nu_j)+\frac{d}{2})}$ into $[S_Y(\bomega)]_{ij}$:
\begin{align*}
&[\bS_Y(\bomega)]_{ij} = \sigma_{ij} 
\frac{\Gamma(\nu_i+\frac{d}{2})^{1/2}}{\Gamma(\nu_i)^{1/2}} \frac{\Gamma(\nu_j+\frac{d}{2})^{1/2}}{\Gamma(\nu_j)^{1/2}} 
\frac{\Gamma(\frac{1}{2}(\nu_i+\nu_j))}{\Gamma(\frac{1}{2}(\nu_i+\nu_j)+\frac{d}{2})}
\frac{\Gamma(\frac{1}{2}(\nu_i+\nu_j)+\frac{d}{2})  \phi^{\nu_i+\nu_j}}{\Gamma(\frac{1}{2}(\nu_i+\nu_j))  \pi^{d/2}} \cdot\\
&\qquad \qquad \qquad \qquad \cdot (\phi^2+\|\bomega\|^2)^{-\frac{\nu_i+\nu_j}{2}-\frac{d}{2}} \\
&= \sigma_{ij} 
\frac{\Gamma(\nu_i+\frac{d}{2})^{1/2}}{\Gamma(\nu_i)^{1/2}} 
\frac{\Gamma(\nu_j+\frac{d}{2})^{1/2}}{\Gamma(\nu_j)^{1/2}} 
\frac{\phi^{\nu_i} \phi^{\nu_j}}{\pi^{d/4}\pi^{d/4}} 
(\phi^2+\|\bomega\|^2)^{-(\frac{\nu_i}{2}+\frac{d}{4})} 
(\phi^2+\|\bomega\|^2)^{-(\frac{\nu_j}{2}+\frac{d}{4})} \\[8pt]
&= \left( \frac{\Gamma(\nu_i+\frac{d}{2})^{1/2} \, \phi^{\nu_i}}{\Gamma(\nu_i)^{1/2} \, \pi^{d/4}} 
(\phi^2+\|\bomega\|^2)^{-(\frac{\nu_i}{2}+\frac{d}{4})} \right) 
\sigma_{ij} \left(
\frac{\Gamma(\nu_j+\frac{d}{2})^{1/2} \, \phi^{\nu_j}}{\Gamma(\nu_j)^{1/2} \, \pi^{d/4}} 
(\phi^2+\|\bomega\|^2)^{-(\frac{\nu_j}{2}+\frac{d}{4})} \right) \\
&= [\bD(\bomega) \, \bSigma \, \bD(\bomega)^*]_{ij},
\end{align*}
where $\bD(\bomega)$ is a diagonal matrix whose $j$th diagonal entry is real-valued and equal to $d_j(\bomega) = \frac{\Gamma(\nu_j + \tfrac{d}{2})^{1/2} \phi^{\nu_j}}{\Gamma(\nu_j)^{1/2} \pi^{d/4}} (\phi^2 + \|\bomega\|^2)^{-(\nu_j/2 + d/4)} = m(\bomega; \nu_j, \phi)^{1/2} = d_j(\bomega)^* $. Therefore, $\bS_Y(\bomega) = \bD(\bomega) \bSigma \bD(\bomega)^*$, and because $\intRd |d_j(\bomega)|^2d\bomega = \intRd m(\bomega; \nu_j, \phi)d\bomega=1$, similarly for $d_i(\bomega)$, the parsimonious Mat\'ern is a spectrally inside-out model. We then compute  
\begin{align*}
c_{ij}(\bl-\bl') &= \intRd e^{\mathrm{i}\bomega\cdot (\bl-\bl')} d_i(\bomega)d_j(\bomega)^*d\bomega \\
&= \intRd e^{\mathrm{i}\bomega\cdot (\bl-\bl')} m(\bomega; \nu_i, \phi)^{1/2} m(\bomega;\nu_j, \phi)^{1/2} d\bomega \\
&= \gamma_{ij} \intRd e^{\mathrm{i}\bomega\cdot (\bl-\bl')} m(\bomega; (\nu_i+\nu_j)/2, \phi) d\bomega = \gamma_{ij} M(\bl-\bl'; (\nu_i+\nu_j)/2, \phi),
\end{align*}
completing the proof. 
\end{proof}

\begin{theorem}\label{thm:iox_supp}
Under the inside-out cross-covariance model for $\bY(\cdot)$, 
\begin{enumerate}
\item[(a)] $\text{corr}\{y_i(\bl), y_j(\bl') \} = \sigma_{ij}/(\sigma_{ii}\sigma_{jj})^{1/2} f_{ij}(\bl, \bl'; \calS)$,
\item[(b)] $\text{corr}\{y_i(\bl), y_j(\bl') \mid \sigY \} = -Q_{ij}/(Q_{ii}Q_{jj})^{1/2} f_{ij}(\bl, \bl'; \calS)$,
\item[(c)] $y_i(\cdot) \indep y_j(\cdot) \mid \bY_{\others}(\cdot)$  if and only if $ Q_{ij} = 0$.
\end{enumerate}

\end{theorem}
\begin{proof}
We find the spatial correlation function immediately, since, by construction of the model, we have
\begin{equation}\label{eq:iox_variance}
\text{var}\{ y_i(\bl) \} = \sigma_{ii} \left[ \bh_{i}(\bl)^T\bL_i \bL_i^T\bh_i(\bl) + \xi_{ii}(\bl, \bl) \right] = \sigma_{ii}\rho_i(\bl,\bl) = \sigma_{ii} \text{ for all } \bl \in \calD,
\end{equation}
where we use the fact that, for all $\bl\in\calD$, we have 
\begin{align*}
\xi_{ii}(\bl,\bl) &= \rho_i(\bl,\bl) - \bh_i(\bl)^T\rho_i(\calS, \bl)\\
&= \rho_i(\bl,\bl) - \rho_i(\bl, \calS)\rho_i(\calS)^{-1} \bL_i\bL_i^T \rho_i(\calS)^{-1}\rho_i(\calS,\bl)\\
&= \rho_i(\bl,\bl) - \bh_i(\bl)^T \bL_i\bL_i^T \bh_i(\bl).
\end{align*}
For the partial cross-correlation function, we cannot rely on spectral results because this cross-covariance model is not stationary. We proceed by operating on four objects: the target process $\bY(\cdot)$, the spatially uncorrelated seed process $\bV(\cdot)$, and the respective ``residual'' processes $\bZ(\cdot)$ and $\bW(\cdot)$. 
Consider the processes $y_i(\cdot)$ and $y_j(\cdot)$ and let $o = \{1, \dots, q \}\setminus \{i,j \}$. 
Denote $\sigV = \sigma\{ \bV_{\others}(\bl), \bl \in \calD \}$, and define $w_j(\bl) = v_j(\bl) - E( v_j(\bl) \mid \sigV )$. 
For all $r\in\{1, \dots, q \}$, $y_r(\bl)$ is a deterministic function of $v_r(\calS)$ and $v_r(\bl)$. Therefore, $\sigY \subset \sigma\{ \bV_{\others}(\calS) \cup \bV_{\others}(\bl), \bl\in \calD \} = \sigV$. For each $j$ we can also write, if $\bl \notin \calS$, $v_j(\bl) = r_j(\bl)^{-1/2} ( y_j(\bl) - \bh_j(\bl)^T y_j(\calS) ) $, and if $\bl \in \calS$ (and hence $r_j(\bl)=0$), we have $\bv_j = \bL_j^{-1} \by_j$, which is well defined by positive definiteness of $\rho_j(\calS) = \bL_j \bL_j^T$. This implies that $v_j(\bl)$ can be expressed as a deterministic function of $y_j(\bl)$ for all $\bl\in\calD$. We then obtain $\sigV \subset \sigY$, concluding that it must be $\sigY = \sigV$. This implies that conditioning on the full path of $\Yothers(\cdot)$ is equivalent to conditioning on the full path of $\Vothers(\cdot)$. From the definition of $z_j(\cdot)$ and $y_j(\cdot)$,
\begin{equation*}
\begin{aligned}
z_j(\bl) &= y_j(\bl) - E(y_j(\bl) \mid \sigY) \\
&= \bh_j(\bl)^T \bL_j \bv_j + r_j(\bl)^{1/2} v_j(\bl) - E(\bh_j(\bl)^T \bL_j \bv_j + r_j(\bl)^{1/2} v_j(\bl) \mid \sigY) \\
&= \bh_j(\bl)^T \bL_j \bv_j + r_j(\bl)^{1/2} v_j(\bl) - E(\bh_j(\bl)^T \bL_j \bv_j + r_j(\bl)^{1/2} v_j(\bl) \mid \sigV) \\
&= \bh_j(\bl)^T \bL_j ( \bv_j - E(\bv_j \mid \sigV ) ) + r_j(\bl)^{1/2} \{ v_j(\bl) - E(v_j(\bl) \mid \sigV) \}\\
&= \bh_j(\bl)^T \bL_j \bw_j + r_j(\bl)^{1/2} w_j(\bl).
\end{aligned}
\end{equation*}


\noindent 
$\bV(\cdot)$ is spatially independent, so $E(v_i(\bl) \mid \sigV) = E(v_i(\bl) \mid \Vothers(\bl))$, $E(v_i(\bl) \mid \sigV) = E(v_i(\bl) \mid \Vothers(\bl))$, and
\begin{align*}
\text{cov}\{ &w_i(\bl), w_j(\bl') \}= \text{cov}\{ v_i(\bl) - E(v_i(\bl) \mid \sigV), v_j(\bl') - E(v_j(\bl') \mid \sigV ) \} \\
&= \text{cov}\{ v_i(\bl) - E(v_i(\bl) \mid \Vothers(\bl)), v_j(\bl') - E(v_j(\bl') \mid \Vothers(\bl') ) \}.
\end{align*}
Again because of spatial independence of $\bV(\cdot)$, $\text{cov}\{ w_i(\bl), w_j(\bl') \}=0$ if $\bl\neq\bl'$. If $\bl=\bl'$, since $\text{cov}\{ \bV(\bl), \bV(\bl) \}=\bSigma$, then for the corresponding residual we have $\text{cov}\{\bW(\bl), \bW(\bl)\} = \bSigma_{x} - \bSigma_{x, \others}\bSigma_{\others}^{-1}\bSigma_{\others, x}$, where $x=\{i,j\}$, $\bSigma_{x}$ is the submatrix of $\bSigma$ where we select $x$ rows and columns, similarly for $\bSigma_{\others}$, and $\bSigma_{x, \others}$ subsets $\bSigma$ to rows $x$ and $\others$ columns. By the properties of block matrix inverses, we find the off-diagonal element of  $\text{cov}\{\bW(\bl), \bW(\bl)\}$ to be
\begin{align*}
    \text{cov}\{ w_i(\bl), w_j(\bl) \} &= \sigma_{ij} - \bSigma_{i, \others} \bSigma_{\others}^{-1} \bSigma_{\others, j}\\
    &= - Q_{ij} / (Q_{ii}Q_{jj} - Q_{ij}^2),
\end{align*}
which implies that $\text{cov}\{ w_i(\bl), w_j(\bl') \}=0$ for all $\bl, \bl' \in \calD \iff Q_{ij}=0$. 
Moving to the residual processes of $\by_i(\cdot)$ and $\by_j(\cdot)$,
\begin{subequations}\label{eq:covz}
\begin{align}
  \text{cov}\{ &z_i(\bl), z_j(\bl') \}= \text{cov}\{ 
\bh_i(\bl)^T \bL_i \bw_i + 
r_i(\bl)^{1/2} w_i(\bl), 
h_j(\bl') \bL_j \bw_j + 
r_j(\bl')^{1/2} w_j(\bl') \} \nonumber\\
&= \bh_i(\bl)^T \bL_i  \text{cov}\{ \bw_i, \bw_j\} \bL_j^T \bh_j(\bl') \label{eq:covz:1}\\
&+ \bh_i(\bl)^T \bL_i  \text{cov}\{\bw_i,  w_j(\bl') \} r_j(\bl')^{1/2} + r_i(\bl)^{1/2} \text{cov}\{ w_i(\bl), \bw_j\} \bL_j^T \bh_j(\bl')  \label{eq:covz:2} \\
&\quad + r_i(\bl)^{1/2}r_j(\bl')^{1/2} \text{cov} \{w_i(\bl), w_j(\bl') \} \label{eq:covz:3}\\
&= -\frac{Q_{ij}}{Q_{ii}Q_{jj} - Q_{ij}^2}\left[ \bh_i(\bl)^T \bL_i \bL_j^T \bh_j(\bl') +  \xi_{ij}(\bl,\bl') \right], \nonumber
\end{align}
\end{subequations}
where, in \eqref{eq:covz:1} we have $\bw_i=w_i(\calS)$, therefore $\text{cov}\{ \bw_i, \bw_j\} = -Q_{ij}/(Q_{ii}Q_{jj} - Q_{ij}^2) \bI_n$. Line \eqref{eq:covz:2} is always zero, since if $\bl'\in\calS$ then $r_i(\bl')=0$, whereas since $\bw_i=w_i(\calS)$, then $\bl'\notin\calS$ gives $\text{cov}\{\bw_i, w_j(\bl')\}=0$, so the first term is always zero and the same logic applies to the second term. The term in \eqref{eq:covz:3} is zero if $\bl\neq\bl'$, and if $\bl=\bl'$, then it equals $-Q_{ij}/(Q_{ii}Q_{jj}-Q_{ij}^2)$.
Using the same logic of \eqref{eq:iox_variance}, we find $\text{var}\{z_i(\bl) \} = Q_{jj}/(Q_{ii}Q_{jj}-Q_{ij}^2)$. This yields the target partial cross-correlation function. 

Finally, because $\bh_i(\bl)^T \bL_i \bL_j^T \bh_j(\bl') +  \xi_{ij}(\bl,\bl')\neq 0$ for some $\bl, \bl' \in \calD$ (e.g., take $\bl=\bl')$, then it must be true that $\text{cov}\{ z_i(\bl), z_j(\bl') \} = 0$ for all $\bl,\bl'\in\calD$ if and only if $Q_{ij}=0$. This concludes the proof since $\text{corr}\{ y_i(\bl), y_j(\bl') \mid \Yothers\} = 0 \iff \text{corr}\{ z_i(\bl), z_j(\bl') \} = 0 \iff \text{cov}\{ z_i(\bl), z_j(\bl') \} = 0 \iff Q_{ij}=0$, and $\text{corr}\{ y_i(\bl), y_j(\bl') \mid \Yothers\} = 0$ defines conditional independence of $y_i(\cdot)$ and $y_j(\cdot)$ given $\Yothers(\cdot)$.
\end{proof}

\begin{theorem}\label{thm:lmc_supp}
Let $\bA=[a_{ij}]=\bLambda^{-1}$. Under a linear model of coregionalization,
\[y_i(\cdot) \indep y_j(\cdot) \mid \bY_{\others}(\cdot) \text{ if and only if } a_{ri}a_{rj}=0 \text{ for all } r=1,\dots,q.\]

\end{theorem}
\begin{proof}
Since $\bS_W(\bomega) 
= \text{diag}\{ \tilde{\rho}_1(\bomega), \dots, \tilde{\rho}_q(\bomega) \}$ is the spectral density matrix of $\bW(\cdot)$, then from $\bY(\bl) = \bLambda \bW(\bl)$ we get $
\bS_Y(\bomega) = \bLambda \bS_W(\bomega) \bLambda^T.$
We assumed $\tilde{\rho}_j(\bomega)>0$ for almost all $\bomega$ and all $j$ and $\rho_i\neq \rho_j$ for all $i\neq j$, hence
\( \bS_Y(\bomega)^{-1} = \bA^{T} \text{diag}\{ 1/\tilde{\rho}_1(\bomega), \dots, 1/\tilde{\rho}_q(\bomega) \}  \bA. \)
Therefore $[\bS_Y(\bomega)^{-1}]_{ij}=\sum_{r=1}^q a_{ri}a_{rj}/\tilde{\rho}_r(\bomega)=0$ if and only if $a_{ri}a_{rj}=0$ for all $r$. This yields process-level conditional independence by Theorem 2.4 of \cite{dahlhaus_graphical_2000}.
\end{proof}

\section{Spectral properties of linear coregionalization}
The spectral density matrix of $\bY(\cdot)$ is $\bS_Y(\bomega) = \bLambda \bR(\bomega) \bLambda^T$, where $\bR(\bomega) = \text{diag}\{ \tilde{\rho}_1(\bomega), \dots, \tilde{\rho}_q(\bomega) \}$. Let $\bL(\bomega) = \bLambda \bR(\bomega)^{1/2}$, so that $\bS_Y(\bomega) = \bL(\bomega) \bL(\bomega)^T$. Partition $\bL(\bomega)$ by rows into the $2\times q$ matrix $\bL_x(\bomega)$ (for rows $i,j$) and the $q-2 \times q$ matrix $\bL_o(\bomega)$ (for rows $o=\{1,\dots, q\}\setminus \{i,j\})$, so the spectral density of $\bZ(\cdot)$ is 
\begin{align*}
\bS_Z(\bomega) 
&= \bL_x(\bomega)\bL_x(\bomega)^T - \bL_x(\bomega) \bL_o(\bomega)^T (\bL_o(\bomega)\bL_o(\bomega)^T)^{-1} \bL_o(\bomega) \bL_x(\bomega)^T\\
&= \bL_x(\bomega) \left[ \bI_q - \bL_o(\bomega)^T(\bL_o(\bomega) \bL_o(\bomega)^T)^{-1}\bL_o(\bomega) \right] \bL_x(\bomega)^T
\end{align*}
Letting $\bP^{\perp}_{o}(\bomega) = \bI_q - \bL_o(\bomega)^T(\bL_o(\bomega) \bL_o(\bomega)^T)^{-1}\bL_o(\bomega)$ be the matrix which project onto the space orthogonal to the columns of $\bL_o(\bomega)$, we get
\begin{align*}
\bS_Z(\bomega) 
&= \bL_x \bP^{\perp}_{o}(\bomega) \bL_x^T,
\end{align*}
which clarifies that the spectral structure of $\bS_Z$ does not straightforwardly factorize into graph-level and spatial components.

\section{Illustration details}

All analyses were run on a workstation equipped with an AMD Ryzen 9 7950X3D 16-core processor, with 192GB of memory, running Ubuntu 22.04, R 4.5.2. For efficient linear algebra, we use BLAS/LAPACK provided in the Intel Math Kernel Library version 2025.3  \citep{lapack99, intel2019mkl}.

\subsection{Parsimonious Mat\'ern}
We fit the parsimonious Mat\'ern via package \texttt{GpGpm}, adapting the source code and scripts available at \url{https://github.com/yf297/GpGp_multi_paper}. For the Vecchia approximation of the Gaussian process, we use $m=20$ nearest neighbors in the ROC comparison, $m=30$ in the simulated data applications, and $m=60$ for the Jura application. We used package \texttt{RhpcBLASctl} to set the number of BLAS threads to 16 before running the model fitting functions.

\subsection{Inside-out cross-covariance}
We fit the inside-out cross-covariance model via package \texttt{spiox}. Specifically, we always fit a ``response'' model, choosing $m=30$ nearest neighbors for the Vecchia approximation of each component of the multivariate Gaussian process. In the simulated data application, we configure the package to sample from the posterior distribution of all parameters, including all spatial correlation parameters (spatial decay $\phi_j$, smoothness $\nu_j$, and proportion of noise $\alpha_j$, for each $j=1,\dots, q$). We compute the signal-to-noise (SNR) ratio as $(1-\alpha)/\alpha$. In the Jura application, we fix the smoothness parameters $\nu_j=\hat{\nu}_j$ for $j=1,\dots,q$, where $\hat{\nu}_j$ is the marginal estimate of smoothness obtained from fitting a Vecchia-approximated univariate Gaussian process via package \texttt{GpGp}; the point estimates are $0.2083, 0.3627, 0.3202, 0.5946, 0.3348, 0.2387, 0.5039$, for \textit{Cd}, \textit{Co}, \textit{Cr}, \textit{Cu}, \textit{Ni}, \textit{Pb}, \textit{Zn}, respectively. The other parameters are estimated via Markov-chain Monte Carlo. Parallel processing was enabled on 16 OMP \citep{dagum1998openmp} threads.

\subsection{Linear model of coregionalization}
We fit the linear model of coregionalization via package \texttt{meshed}, available at \url{github.com/mkln/meshed}. We use blocks of size \texttt{block\_size}$=60$ and $k=q$ latent factors for the underlying meshed Gaussian process approximation \citep{peruzzi_spatial_2024}. Parallel processing was enabled on 16 OMP \citep{dagum1998openmp} threads.

\end{document}